\documentclass[runningheads]{llncs}
\usepackage{fullpage}
\usepackage{graphicx}
\usepackage[usenames,dvipsnames]{xcolor}
\usepackage{makeidx}
\usepackage{algorithm}
\usepackage{algorithmic}
\usepackage{graphicx,tipa}
\usepackage{arcs,lmodern,fix-cm}
\usepackage{times}
\usepackage{amsmath,amsfonts}
\usepackage{slashbox,multirow}
\usepackage{rotating}
\usepackage{verbatim}
\usepackage{caption}
\usepackage{comment}
\usepackage{caption}
\usepackage{subcaption}

\usepackage{fourier}

\newcommand{\reals}{\mathbb{R}}

\newcommand{\ignore}[1]{}

\newcommand{\ddd}{\mathrm{d}}

\newcommand{\norm}[1]{\left\lVert#1\right\rVert}

\def\qed{\hfill\rule{2mm}{2mm}}

\newcommand{\arccc}[1]{
\widearc{#1}
}
\newcommand{\barrrr}[1]{
\overline{#1}
}

\excludecomment{showproof}

\begin{document}
\title{Evacuating from $\ell_p$ Unit Disks in the Wireless Model
\thanks{
This is the full version of the paper with the same title which will appear in the proceedings of the 
17th International Symposium on Algorithms and Experiments for Wireless Sensor Networks (ALGOSENSORS 2021), September 9-10, 2021, Lisbon, Portugal.
}
}
%
%
\author{
Konstantinos Georgiou
\thanks{Research supported in part by NSERC.}
\and
Sean Leizerovich
\and
Jesse Lucier
\thanks{Research supported by a NSERC USRA.}
\and
Somnath Kundu
\ignore{
First Author\inst{1}
\and
Second Author\inst{2,3}
\and
Third Author\inst{3}
} 
}
\authorrunning{K. Georgiou et al.}
%
\institute{
Department of Mathematics, Ryerson University,
Toronto, ON, M5B 2K3, Canada
\email{\{konstantinos,sleizerovich,jesse.lucier,somnath.kundu\}@ryerson.ca}
}
\maketitle              
\begin{abstract}
The search-type problem of evacuating 2 robots in the wireless model from the (Euclidean) unit disk was first introduced and studied by Czyzowicz et al. [DISC'2014]. Since then, the problem has seen a long list of follow-up results pertaining to variations as well as to upper and lower bound improvements. All established results in the area study this 2-dimensional search-type problem in the Euclidean metric space where the search space, i.e. the unit disk, enjoys significant (metric) symmetries. 

We initiate and study the problem of evacuating 2 robots in the wireless model from $\ell_p$ unit disks, $p \in [1,\infty)$, where in particular robots' moves are measured in the underlying metric space. 
To the best of our knowledge, this is the first study of a search-type problem with mobile agents in more general metric spaces. 
The problem is particularly challenging since even the circumference of the $\ell_p$ unit disks have been the subject of technical studies. 
In our main result, and after identifying and utilizing the very few symmetries of $\ell_p$ unit disks, we design \emph{optimal evacuation algorithms} that vary with $p$. Our main technical contributions are two-fold. 
First, in our upper bound results, we provide (nearly) closed formulae for the worst case cost of our algorithms. 
Second, and most importantly, our lower bounds' arguments reduce to a novel observation in convex geometry which analyzes trade-offs between arc and chord lengths of $\ell_p$ unit disks as the endpoints of the arcs (chords) change position around the perimeter of the disk, which we believe is interesting in its own right. 
Part of our argument pertaining to the latter property relies on a computer assisted numerical verification that can be done for non-extreme values of $p$.

\keywords{
Search
\and
Evacuation
\and
Wireless model
\and 
$\ell_p$ metric space
\and
Convex and computational geometry}
\end{abstract}

\section{Introduction}

In the realm of mobile agent computing, search-type problems are concerned with the design of searchers' (robots') trajectories in some known search space to locate a hidden object. 
Single searcher problems have been introduced and studied as early as the 60's by the mathematics community~\cite{beck1964linear,bellman1963optimal}, and later in the late 80's and early 90's by the theoretical computer science community~\cite{baezayates1993searching}. 
The previously studied variations focused mainly on the type of search domain, e.g. line or plane or a graph, and the type of computation, e.g. deterministic or randomized. 
Since search was also conducted primarily by single searchers, termination was defined as the first time the searcher hit the hidden object.
In the last decade with the advent of robotics,
search-type problems have been rejuvenated within the theoretical computer science community, which is now concerned with novel variations including the number of searchers (mobile agents),the communication model, e.g. face-to-face or wireless, and robots' specifications, e.g. speeds or faults, including crash-faults or byzantine faults. 
As a result of the multi-searcher setup, termination criteria are now subject to variations too, and these include the number or the type of searchers that need to reach the hidden item (for a more extended discussion with proper citations, see Section~\ref{sec: related work}). 

One of the most studied search domains, along with the line, is that of a circle, or a disk. In a typical search-type problem in the disk, the hidden item is located on the perimeter of the unit circle, and searchers start in its center. Depending on the variation considered, and combining all specs mentioned above, a number of ingenious search trajectories have been considered, often with counter-intuitive properties. Alongside the hunt for upper bounds (as the objective is always to minimize some form of cost, e.g. time or traversed space or energy) comes also the study of lower bounds, which are traditionally much more challenging to prove (and which rarely match the best known positive results). 

Search on the unbounded plane as well as in other 2-dimensional domains, e.g. triangles or squares, has been considered too, giving rise to a long list of treatments, often with fewer tight (optimal) results. 
While the list of variations for searching on the plane keeps growing, there is one attribute that is common to all previous results where robots' trajectories lie in $\reals^2$, which is the underlying Euclidean metric space. In other words, distances and trajectory lengths are all measured with respect to the Euclidean $\ell_2$ norm. Not only the underlying geometric space is well understood, but it also enjoys symmetries, and admits standard and elementary analytic tools from trigonometry, calculus, and analytic geometry. 

We deviate from previous results, and to the best of our knowledge, we initiate the study of a search-type problem with mobile agents in $\reals^2$ where the underlying metric space is induced by any $\ell_p$ norm, $p\geq 1$. 
The problem is particularly challenging since even ``highly symmetric'' shapes, such as the unit circle, enjoy fewer symmetries in non-Euclidean spaces. 
Even more, robot trajectories are measured with respect to the underlying metric, giving rise to technical mathematical expressions for measuring the performance of an algorithm.
In particular, we consider the problem of reaching (evacuating from) a hidden object (the exit) placed on the perimeter of the $\ell_p$ unit circle. Our unit-speed searchers start from the center of the circle, placed at the origin of the Cartesian plane $\reals^2$, and are controlled by a centralized algorithm that allows them to communicate their findings instantaneously. Termination is determined by the moment that the last searcher reaches the exit, and the performance analysis is evaluated against a deterministic worst case adversary. For this problem we provide \emph{optimal evacuation algorithms}. Apart from the novelty of the problem, our contributions pertain to (a) a technical analysis of search (optimal) algorithms that have to vary with $p$, giving rise to our upper bounds, and to
(b) an involved geometric argument that also uses, to the best of our knowledge, a novel observation on convex geometry that relates a given $\ell_p$ unit circle's arcs to its chords, giving rise to our matching lower bounds.

\subsection{Related Work}
\label{sec: related work}

Our contributions make progress in Search-Theory, a term that was coined after several decades of celebrated results in the area, and which have been summarized in books~\cite{ahlswede1987search,Alpern2013,alpern2002theory,stone1975theory}. The main focus in that area pertains to the study of (optimal) searchers' strategies who compete against (possibly hidden) hider(s) in some search domain. An even wider family of similar problems relates to exploration~\cite{AKS02}, terrain mapping,~\cite{mitchell2000geometric}, and hide-and-seek and pursuit-evasion~\cite{nahin2012chases}. 

The traditional problem of searching with one robot on the line~\cite{baezayates1993searching} has been generalized with respect to 
the number of searchers, 
the type of searchers, 
the search domain, 
and the objective, among others. 
When there are multiple searchers and the objective is that all of them reach the hidden object, the problem is called an \emph{evacuation problem}, with the first treatments dating back to over a decade ago~\cite{baumann2009earliest,FGK10}.
The evacuation problem that we study is a generalization of a problem introduced by Czyzowicz et al.~\cite{CGGKMP} and that was solved optimally. In that problem, a hidden item is placed on the (Euclidean) unit disk, and is to be reached by two searchers that communicate their findings instantaneously (wireless model). 
Variations of the problem with multiple searchers, as well as of another communication model (face-to-face) was considered too, giving rise to a series of follow-up papers~\cite{Watten2017,CGKNOV,disser2019evacuating}. 
Searching the boundary of the disc is also relevant to so-called Ruckle-type games, and closely related to our problem is a variation mentioned in~\cite{baston2013some} as an open problem, in which the underlying metric space is any $\ell_p$-induced space, $p\geq 1$, as in our work.

The search domain of the unit circle that we consider is maybe one of the most well studied, together with the line~\cite{Groupsearch}. 
Other topologies that have been considered include
multi-rays~\cite{BrandtFRW20},
triangles~\cite{ChuangpishitMNO17,CzyzowiczKKNOS15},
and graphs~\cite{angelopoulos2019expanding,Borowiecki0DK16}. 
Search for a hidden object on an unbounded plane was studied in~\cite{LS01}, later in~\cite{Emekicalp2014,Lenzen2014}, and more recently in~\cite{AcharjeeGKS19,DKP20}.

Search and evacuation problems with faulty robots have been studied in 
~\cite{czyzowicz2017evacuation,GKLPP19Algosensors,pattanayak2019chauffeuring} and with probabilistically faulty robots in~\cite{bgmp2020probabilistically}. 
Variations pertaining to the searcher's speeds appeared in~\cite{GeorgiouKK16,dmtcs:5528} (immobile agents), in~\cite{lamprou2016fast} (speed bounds) and in~\cite{CzyzowiczKKNOS17} (terrain dependent speeds). 
Search for multiple exits was considered in~\cite{czyzowicz2018evacuating,PattanayakR0S18},
while variation of searching with advice appeared in~\cite{georgiou2017searching}. 
Some variations of the objective include the so-called priority evacuation problem~\cite{CGKKKNOS18b,CzyzowiczGKKKNO20} and its generalization of weighted searchers~\cite{GL20}. 
Randomized search strategies have been considered in~\cite{beck1964linear,bellman1963optimal} and later in~\cite{kao1996searching}
for the line, and more recently in~\cite{ChuangpishitGS18} for the disk. 
Finally, turning costs have been studied in~\cite{demaine2006online} and an objective of minimizing a notion pertaining to energy (instead of time) was studied in~\cite{czyzowiczICALP2019,kranakis2009time}, just to name a few of the developments related to our problem. The reader may also see recent survey~\cite{CGK19search} that elaborates more on selected topics.

\subsection{High Level of New Contributions \& Motivation}
\label{sec: new contributions}

The algorithmic problem of searching in arbitrary metric spaces has a long history~\cite{chavez2001searching}, but the focus has been mainly touching on database management.
In our work, we extend results of a search-type problem in mobile agent computing first appeared in~\cite{CGGKMP}. More specifically, we provide optimal algorithms 
for the search-type problem of evacuating two robots in the wireless model from the $\ell_p$ unit disk, for $p\geq 1$ (previously considered only for the Euclidean space $p=2$). 
The novelty of our results is multi-fold. First, to the best of our knowledge, this is the first result in mobile agent computing in which a search problem is studied and optimally solved in $\ell_p$ metric spaces. Second, both our upper and lower bound arguments rely on technical arguments. Third, part of our lower bound argument relies on an interesting property of unit circles in convex geometry, which we believe is interesting in its own right. 

The algorithm we prove to be optimal for our evacuation problem is very simple, but it is one among infinitely many natural options one has to consider for the underlying problem (one for each deployment point of the searchers). Which of them is optimal is far from obvious, and the proof of optimality is, as we indicate, quite technical.

Part of the technical difficulty of our arguments 
arises from the implicit integral expression of arc lengths of $\ell_p$ circles. 
Still, by invoking the Fundamental Theorem of Calculus we determine the worst case placement of the hidden object for our algorithms.
Another significant challenge of our search problem pertains to the limited symmetries of the unit circle in the underlying metric space. As a result, it is not surprising that the behaviour of the provably optimal algorithm does depend on $p$, with $p=2$ serving as a threshold value for deciding which among two types of special algorithms is optimal.  
Indeed, consider an arbitrary contiguous arc of some fixed length of the $\ell_p$ unit circle with endpoints $A,B$. In the Euclidean space, i.e. when $p=2$, the length of the corresponding chord is invariant of the locations of $A,B$. In contrast, for the unit circle hosted in any other $\ell_p$ space, the slope of the chord $AB$ does determine its length. The relation to search and evacuation is that the arc corresponds to a subset of the search domain which is already searched, and points $A,B$ are the locations of the searchers when the exit is reported. Since searchers operate in the wireless model in our problem (hence one searcher will move directly to the other searcher when the hidden object is found), 
their trajectories are calculated 
so that their $\ell_p$ distance is the minimum possible for the same elapsed search time. 

Coming back to the $\ell_p$ unit disks, we show an interesting property which may be of independent interest (and which we did not find in the current literature). More specifically, and in part using computer assisted numerical calculations for a wide range of values of $p$, we show that for any arc of fixed length, the placement of its endpoint $A,B$ that minimizes the $\ell_p$ length of chord $AB$ is when $AB$ is parallel to the $y=0$ or $x=0$ lines, for $p\leq 2$, and when $AB$ is parallel to the $y=x$ or $y=-x$ lines for $p\geq 2$.  
The previous fact is coupled by a technical extension of a result first sketched in~\cite{CGGKMP}, according to which at a high level, as long as searchers have left any part of the unit circle of cumulative length $\alpha$ unexplored (not necessarily contiguous), then there are at least two unexplored points of arc distance \emph{at least} $\alpha$.


\section{Problem Definition, Notation \& Nomenclature}

For a vector $x=(x_1,x_2) \in \reals^2$, we denote by $\norm{x}_p$ the vector's $\ell_p$ norm, i.e. $\norm{x}_p=\left(|x_1|^p+|x_2|^p\right)^{1/p}$. The \emph{$\ell_p$ unit circle} is defined as 
$
\mathcal C_p := \left\{ x\in \reals^2: \norm{x}_p=1\right\},
$
see also Figure~\ref{fig: UnitCircles} for an illustration.
We equip $\reals^2$ with the metric $d_p$ induced by the $\ell_p$ norm, i.e. for $x,y \in \reals^2$ we write 
$
d_p(x,y) = \norm{x-y}_p. 
$
Similarly, if $r:[0,1] \mapsto \reals^2$ is an injective and continuously differentiable function, it's \emph{$\ell_p$ length} is defined as 
$
\mu_p(r):= \int_0^1 \norm{r'(t)}_p \ddd t.
$
As a result, a unit speed robot can traverse $r([0,1])$ in metric space $(\reals^2,d_p)$ in time $\mu_p(r)$.

We proceed with a formal definition of our search-type problem. 
In problem \textsc{WE}$_p$ (\emph{Wireless Evacuation in $\ell_p$ space}, $p\geq 1$), two unit-speed robots start at the center of a unit circle $\mathcal C_p$ placed at the origin of the metric space $(\reals^2,d_p)$.
Robots can move anywhere in the metric space, and they operate according to a centralized algorithm. 
An \emph{exit} is a point $P$ on the perimeter of $\mathcal C_p$.
An \emph{evacuation algorithm} $A$ consists of robots trajectories, 
either of which may depend on the placement of $P$ only after at least one of the robots 
passes through $P$ (\emph{wireless model}).\footnote{An underlying assumption is also that robots can distinguish points $(x,y)$ by their coordinates, and they can move between them at will. As a byproduct, robots have a sense of orientation. This specification was not mentioned explicitly before for the Euclidean space, since all arguments were invariant under rotations (which is not the case any more). However, even in the $\ell_2$ case this specification was silently assumed by fixing the cost of the optimal offline algorithm to 1 (a searcher that knows the location of the exit goes directly there), hence all previous results were performing competitive analysis by just doing worst case analysis.}
For each exit $P$, we define the evacuation cost of the algorithm as the first instance that the last robot reaches $P$. The \emph{cost of algorithm $A$} is defined as the supremum, over all placements $P$ of the exit, of the evacuation time of $A$ with exit placement $P$.
Finally, the \emph{optimal evacuation cost of \textsc{WE}$_p$} is defined as the infimum, over all evacuation algorithms $A$, of the cost of $A$.

Next we show that $\mathcal C_p$ has 4 axes of symmetry (and of course $\mathcal C_2$ has infinitely many, i.e. any line $ax+by=0, a,b \in \reals$). \begin{lemma}
\label{lem: symmetries}
Lines $y=0, x=0, y=x, y=-x$ are all axes of symmetry of $\mathcal C_p$.
Moreover, the center of $\mathcal C_p$ is its point of symmetry. 
\end{lemma}

\begin{proof}
Reflection of point $P=(a,b)$ across lines $y=0, x=0, y=x, y=-x$ give points 
$P_1=(a,-b), P_2=(-a,b), P_3=(b,a), P_4=(-b,-a)$, respectively. It is easy to see that setting $\norm{P}_p=1$ implies that $\norm{P_i}_p=1, i=1,2,3,4$. 
\qed \end{proof}

We use the generalized trigonometric functions $\sin_p(\cdot), \cos_p(\cdot)$, as in \cite{richter2007generalized}, which are defined as
$
\sin_p(\phi) := \sin(\phi)/N_p(\phi), ~~
\cos_p(\phi) := \cos(\phi)/N_p(\phi),
$
where 
$
N_p(\phi) := \left( |\sin(\phi)|^p + |\cos(\phi)|^p \right)^{1/p}. 
$
By introducing 
$$
\rho_p(\phi) := \left( \cos_p(\phi), \sin_p(\phi)  \right), 
$$
which is injective and continuously differentiable function in each of the 4 quadrants, we have the following convenient parametric description of the $\ell_p$ unit circle; 
$
\mathcal C_p = \{ \rho_p (\phi): \phi \in [0,2\pi) \}.
$
In particular, set $Q_1=[0,\pi/2), Q_2=[\pi/2,\pi), Q_3=[\pi, 3\pi/2), Q_4=[3\pi/2,2\pi)$, and define for each $U \subseteq \mathcal C_p$ it's \emph{length} (measure) as
$$
\mu_p(U)
=
\sum_{i=1}^4 
\int_{t \in Q_i: \rho_p(t) \in U} \norm{\rho_p'(t)}_p \ddd t.
$$
It is easy to see that $\mu_p(\cdot)$ is indeed a measure, hence it satisfies the principle of inclusion-exclusion over $\mathcal C_p$.
Also, by Lemma~\ref{lem: symmetries} it is immediate that for every $U \subseteq \mathcal C_p$, and for $\overline U=\{\rho_p(t+\pi): \rho(t) \in U\}$, we have that $\mu_p(U) = \mu_p(\overline U)$ (both observations will be used later in Lemma~\ref{lem: arc chord points}). As a corollary of the same lemma, we also formalize the following observation. 
\begin{lemma}
\label{lem: symmetric measures}
For any $\phi \in \{k\cdot \pi/4: k=0,1,2,3,4\}$ and $\theta \in [0,\pi]$, 
let 
$U_+=\{\rho_p(\phi+t): t\in [0,\theta]\}$
and
$U_-=\{\rho_p(\phi-t): t\in [0,\theta]\}$. Then, we have that $\mu_p(U_+) = \mu_p(U_-)$. 
\end{lemma}

The perimeter of the $\ell_p$ unit circle can be computed as 
$$
\mu_p(\mathcal C_p) = 
\sum_{i=1}^4 
\int_{Q_i} \norm{\rho_p'(t)}_p \ddd t = 
4
\int_{0}^{\pi/2} \norm{\rho_p'(t)}_p \ddd t
:=
2\pi_p. 
$$
By Lemma~\ref{lem: symmetric measures}, we also have 
$
\int_0^{\pi/2} \norm{\rho_p'(t)}_p \ddd t = 2\int_0^{\pi/4} \norm{\rho_p'(t)}_p \ddd t
=
\pi_p/2.
$
Clearly $\mu_2(\mathcal C_2)/2= \pi_2 = \pi=3.14159\ldots$, while the rest of the values of $\pi_p$, for $p\geq 1$, do not have known number representation, in general. 
However, it is easy to see that $\pi_1= \pi_\infty=4$. More generally we have that $\pi_p=\pi_q$ whenever $p,q\geq 1$ satisfy $1/p+1/q=1$~\cite{Keller:2009:V}. As expected, $\pi_2=\pi$ is also the minimum value of $\pi_p$, over $p\geq 1$~\cite{AdlTan00}, see also Figure~\ref{fig: ValueOfPi} for the behavior of $\pi_p$.

For every $\phi, \theta \in [0,2\pi)$, let $A=\rho(\phi), B=\rho(\phi+\theta)$ be two points on the $\ell_p$ unit circle. 
The \emph{chord $\overline{AB}$} is defined as the line segment with endpoints $A,B$. From the previous discussion we have $\mu_p\left( \overline{AB} \right) = d_p(A,B)$.
The \emph{arc $\arccc{AB}$} is defined as the curve $\{\rho(\phi+t): t\in [0,\theta]\}$, hence arcs identified by their endpoints are read counter-clockwise. 
	The length of the same arc is computed as $\mu_p\left( \arccc{AB} \right)$. 

Finally, the \emph{arc distance} of two points $A,B\in \mathcal C_p$ is defined as 
$
\arccc d_p (A,B) := \min\left\{ 
\mu_p\left( \arccc{AB} \right), \mu_p\left( \arccc{BA} \right)  \right\},
$
which can be shown to be a metric. By definition, it follows that $\arccc d_p (A,B) \in [0, \pi_p]$.

Next we present an alternative parameterization of the $\ell_p$ unit circle that will be convenient for some of our proofs. We define 
\begin{equation}
\label{equa: alternative param}
r_p(s) := \left( -s, \left(1- |s|^p\right)^{1/p} \right),
\end{equation}
and we observe that $r_p(s) \in \mathcal C_p$, for every $s \in [-1,1]$. 
It is easy to see that as $s$ ranges from $-1$ to $1$, we traverse the upper 2 quadrants of the unit circle with the same direction as $\rho_p(t)$, when $t$ ranges from $0$ to $\pi$. 
Moreover, for every $t\in [0,\pi]$, there exists unique $s=s(t)$, with $s\in [-1,1]$ such that $\rho_p(t)=r_p(s)$, and $s(t)$ strictly increasing in $t$ with $s(0)=-1, s(\pi/4)=-2^{-1/p}, s(\pi/2)=0, s(3\pi/4)=2^{-1/p}$ and $s(\pi)=1$.


\section{Algorithms for Evacuating 2 Robots in $\ell_p$ Spaces}

First we present a family of algorithms Wireless-Search$_p$($\phi$) for evacuating 2 robots from the $\ell_p$ unit circle $\mathcal C_p$. The family is parameterized by $\phi \in \reals$, see also Figure~\ref{fig: SearchingTwoCircles} for two examples, Algorithm Wireless-Search$_{1.5}$(0) and Wireless-Search$_{3}$($\pi/4$).
\begin{algorithm}
\caption{Wireless-Search$_p$($\phi$)}
\label{algo: wireless search phi}
\begin{algorithmic}[1]
\STATE Both robots move to point $\rho_p(\phi)$. 
\STATE Robots follow trajectories $\rho_p(\phi\pm t)$, $t\geq 0$, till the exit is found and communicated. 
\STATE Finder stays put, and non-finder moves to finder's location along the shortest chord (line segment). 
\end{algorithmic}
\end{algorithm}

Our goal is to prove the following. 
\begin{theorem}
\label{thm: upper bound best algo}~\\
For all $p\in [1,2]$, Algorithm Wireless-Search$_p$(0) is optimal. \\
For all $p\in [2,\infty)$, Algorithm Wireless-Search$_p$($\pi/4$) is optimal.
\end{theorem}
Figure~\ref{fig: AlgoPerformance} depicts the performance of our algorithms as $p\geq 1$  varies. 
Our analysis is formal, however we do rely on computer-assisted numerical calculations to verify certain analytical properties in convex geometry (see proof of Lemma~\ref{lem: mathcal L monotonicity} on page~\pageref{lem: mathcal L monotonicity}, and proof of Lemma~\ref{lem: min arc is what I used} on page~\pageref{lem: min arc is what I used}) that effectively contribute a part of our lower bound argument for bounded values of $p$, as well as $p=\infty$. For large values of $p$, e.g. $p\geq 1000$, where numerical verification is of limited help, we provide provable upper and lower bounds that differ by less than $0.042\%$, multiplicatively (or less than $0.0021$, additively).

Recall that as $\phi$ ranges in $[0,2\pi)$, then $\rho_p(\phi)$ ranges over the perimeter of $\mathcal C_p$. In particular, for any execution of Algorithm~\ref{algo: wireless search phi}, the exit will be reported at some point $\rho_p(\phi \pm t)$, where $t \in [0, \pi]$. 
Since in the last step of the algorithm, the non-finder has to traverse the line segment defined by the locations of robots when the exit is found, we may assume without loss of generality that the exit is always found at some point $\rho_p(\phi \pm t)$, where $t \in [0, \pi]$, say by robot \#1. 
Note that even though Algorithm~\ref{algo: wireless search phi} is well defined for all $[0,2\pi)$ (in fact all reals), due to Lemma~\ref{lem: symmetries} it is enough to restrict to $\phi \in [0, \pi/4]$.

In the remaining of this section, we denote by $\mathcal E_{p,\phi} (\tau)$ the evacuation time of Algorithm Wireless-Search$_p$($\phi$), given that the exit is reported after robots have spent time $\tau$ searching in parallel. We also denote by $\delta_{p,\phi}(\tau)$ the distance of the two robots at the same moment, assuming that no exit has been reported previously. Hence, 
\begin{equation}
\label{equa: evacuation time at tau}
\mathcal E_{p,\phi} (\tau) = 1 +\tau + \delta_{p,\phi}(\tau).
\end{equation}
Since $\mu_p(\mathcal C_p) = 2\pi_p$ and the two robots search in parallel, an exit will be reported for some $\tau \in [0,\pi_p]$. 
Hence, the worst case evacuation time $E_{p,\phi}$ of Algorithm Wireless-Search$_p$($\phi$) is given by 
$$
E_{p,\phi}:= \max_{\tau \in [0,\pi_p]} \mathcal E_{p,\phi} (\tau).\footnote{For arbitrary algorithms one should define the cost as the supremum over all exit placements. Since in Algorithm Wireless-Search$_p$($\phi$) the searched space remains contiguous and its boundaries keep expanding with time, the maximum always exists.}
$$

\subsection{Worst Case Analysis of Algorithm Wireless-Search$_p$($\phi$)}

It is important to stress that parameter $t$ in the description of robots' trajectories in Algorithm Wireless-Search$_p$($\phi$) does not represent the total elapsed search time. Even more, and for an arbitrary value of $\phi$, it is not true that robots occupy points $\rho_p(\phi\pm t)$ simultaneously. To see why, recall that from the moment robots deploy to point $\rho_p(\phi)$, they need time
$
\alpha_{1,2}(\phi, t)
:=
\mu_p\left( 
\left\{ 
\rho_p(\phi \pm s):~s\in [0,t]
\right\}
\right)
$
in order to reach points $\rho_p(\phi\pm t)$. Moreover, $\alpha_{1}(\phi, t) \not = \alpha_{2}(\phi, t)$, unless $\phi = k\cdot \pi/4$ for some $k=0,1,2,3$, as per Lemma~\ref{lem: symmetric measures}. We summarize our observation with a lemma. 

\begin{lemma}
\label{lem: sym trajectories}
Let $\phi \in \{0,\pi/4\}$, and consider an execution of Algorithm~\ref{algo: wireless search phi}.
When one robot is located at point $\rho_p(\phi+t)$, for some $t\in [-\pi, \pi]$, then the other robot is located $\rho_p(\phi-t)$, and in particular $\alpha_{1}(\phi, t) = \alpha_{2}(\phi, t)$.
\end{lemma}

Now we provide worst case analysis of two Algorithms for two special cases of metric spaces. The proof is a warm-up for the more advanced argument we employ later to analyze arbitrary metric spaces.

\begin{lemma}
\label{lem: upper bound 1 and infty}
$E_{1,0}=E_{\infty,\pi/4}=5.$
\end{lemma}

\begin{proof}
First we study Algorithm Wireless-Search$_1$($0$) for evacuating 2 robots from the $\ell_1$ unit disk. By~\eqref{equa: evacuation time at tau}, if the exit is reported after time $\tau$ of parallel search, then 
$
\mathcal E_{1,0} (\tau) = 1 +\tau + \delta_{1,0}(\tau).
$
Note that $\pi_1=4$, so the exit is reported no later than parallel search time $4$. 
First we argue that $\mathcal E_{1,0} (\tau)$ is increasing for $\tau \in [0,2]$.
Indeed, in that time window robot \#1 is moving from point $(1,0)$ to point $(0,1)$ along trajectory $(1-\tau/2,\tau/2)$ (note that this parameterization induces speed 1 movement). 
By Lemma~\ref{lem: sym trajectories}, robot \#2 at the same time is at point $(1-\tau/2,-\tau/2)$. It follows that $\delta_{1,0}(\tau)=\tau$, so indeed $\mathcal E_{1,0} (\tau)$ is increasing for $\tau \in [0,2]$. 
Finally we show that $\mathcal E_{1,0} (\tau)=5$, for all $\tau \in [2,4]$. Indeed, note that for the latter time window, robot \#1 moves from point $(0,1)$ to point $(-1,0)$ along trajectory $(-(\tau-2)/2,1-(\tau-2)/2)$. By Lemma~\ref{lem: sym trajectories}, robot \#2 at the same time is at point $(-(\tau-2)/2,-1+(\tau-2)/2)$. It follows that $\delta_{1,0}(\tau)=|-(\tau-2)/2+(\tau-2)/2|+|1-(\tau-2)/2+1-(\tau-2)/2|=4-\tau$, and hence 
$\mathcal E_{1,0} (\tau) = 1+\tau + \delta_{1,0}(\tau) = 5$, as wanted.

Next we study Algorithm Wireless-Search$_\infty$($\pi/4$) for evacuating 2 robots from the $\ell_\infty$ unit disk. By~\eqref{equa: evacuation time at tau}, if the exit is reported after time $\tau$ of parallel search, then 
$
\mathcal E_{\infty,\pi/4} (\tau) = 1 +\tau + \delta_{\infty,\pi/4}(\tau).
$
As before, $\pi_\infty=4$, so the exit is reported no later than parallel search time $4$. 
We show again that $\mathcal E_{\infty,\pi/4} (\tau)$ is increasing for $\tau \in [0,2]$.
Indeed, in that time window robot \#1 is moving from point $(1,1)$ to point $(-1,1)$ along trajectory $(1-\tau,1)$ (note that this induces speed 1 movement). 
By Lemma~\ref{lem: sym trajectories}, robot \#2 at the same time is at point $(1,1-\tau)$. It follows that $\delta_{\infty,\pi/4}(\tau)=\max\{|1-\tau-1|,|1-1+\tau|\}=\tau$, 
so indeed $\mathcal E_{\infty,\pi/4} (\tau)$ is increasing for $\tau \in [0,2]$. 
Finally we show that $\mathcal E_{\infty,\pi/4} (\tau)=5$, for all $\tau \in [2,4]$. Indeed, note that for the latter time window, robot \#1 moves from point $(-1,1)$ to point $(-1,-1)$ along trajectory $(-1,1-(\tau-2))$. By Lemma~\ref{lem: sym trajectories}, robot \#2 at the same time is at point $(1-(\tau-2),-1)$. It follows that $\delta_{\infty,\pi/4}(\tau)=\max\{|-1-1+(\tau-2)|,|1-(\tau-2)+1|\}=4-\tau$, and hence 
$\mathcal E_{\infty,\pi/4} (\tau) = 1+\tau + \delta_{\infty,\pi/4}(\tau) = 5$, as wanted.  
\qed \end{proof}

It is interesting to see that the algorithms of Lemma~\ref{lem: upper bound 1 and infty} outperform algorithms with different choices of $\phi$. For example, it is easy to see that $E_{1,\pi/4}\geq 6$. Indeed, note that Algorithm Wireless-Search$_1$($\pi/4$) deploys robots at point $(1/2,1/2)$. 
Robot reaches point $(0,1)$ after 1 unit of time, and it reaches point $(-1,0)$ after an additional 2 units of time. The other robot is then at point $(0,-1)$, at an $\ell_1$ distance of 2. So, the placement of the exit at point $(-1,0)$ induces cost $1+1+2+2=6$. A similar argument shows that $E_{\infty,0}\geq 6$ too.


We conclude this section with a summary of our positive results, introducing at the same time some useful notation. The technical and lengthy proof can be found in Appendix~\ref{sec: proof of positive results}.

\begin{theorem}
\label{thm: explored optimal algo}
Let $w_p$ be the unique root to equation $w^p+1=2(1-w)^p$, and define

$$
s_p :=
\left\{
\begin{array}{ll}
 \left(\left(2^p-1\right)^{\frac{1}{p-1}}+1\right)^{-1/p} &, ~~p\in (1,2] \\
\left(w_p^{p/(p-1)} +1\right)^{-1/p} &,~~ p \in (2,\infty) .
\end{array}
\right.
$$

For every $p \in (1,2]$, the placement of the exit inducing worst case cost for Algorithm Wireless-Search$_p$(0) results in the total explored portion of $\mathcal C_p$ with measure
$$e^-_p:=
\pi_p+2\int_{0}^{s_p} \left(z^{p^2-p} \left(1-z^p\right)^{1-p}+1\right)^{1/p} \ddd z.$$
Also, when the exit is reported, robots are at distance 
$\gamma^-_p:=2 (1-s_p^p)^{1/p}$. 

For every $p \in [2,\infty)$, the placement of the exit inducing worst case cost for Algorithm Wireless-Search$_p$($\pi/4$) results in the total explored portion of $\mathcal C_p$ with measure
$$e^+_p:=
\pi_p+2\int_{2^{-1/p}}^{s_p} \left(z^{p^2-p} \left(1-z^p\right)^{1-p}+1\right)^{1/p} \ddd z
.$$
Also, when the exit is reported, robots are at distance 
$\gamma^+_p:= 2^{1/p} \left( 
\left(1-s_p^p\right)^{1/p}
+s_p
\right)
$.

We also set $e_p$ (and $\gamma_p$) to be equal to $e_p^-$ (and $\gamma_p^-$) if $p\leq 2$, and equal to $e_p^+$ (and $\gamma_p^+$) if $p> 2$, and in particular $e_p \in (\pi_p, 2\pi_p]$.
\end{theorem}

Quantities $e_p, \gamma_p$, and some of their properties are depicted in Figures~\ref{fig: GammapEp},~\ref{fig: RelativeEp}, and discussed in Section~\ref{sec: visualizations}. One can also verify that 
$\lim_{p\rightarrow 2^-} e_p^- = \lim_{p\rightarrow 2^+} e_p^+=4\pi/3$, and that $\lim_{p\rightarrow 2^-} \gamma_p^- = \lim_{p\rightarrow 2^+} \gamma_p^+=\sqrt3$.
In order to justify that indeed $e_p \in (\pi_p, 2\pi_p]$, recall that by Lemma~\ref{equa: distance robots} robots' positions during the first $\pi_p/2$ search time (after robots reach perimeter in time 1) is an increasing function. Since the rate of change of time is constant (it remains strictly increasing) for the duration of the algorithm, it follows that the evacuation cost of our algorithms remains increasing till some additional search time. Since robots search in parallel and in different parts of $\mathcal C_p$, and since $e_p$ is the measure of the combined explored portion of the unit circle, it follows that for $e_p>2\pi_p/2=\pi_p$. At the same time, the unit circle has circumference $2\pi_p$, hence $e_p\leq 2\pi_p$.

In other words, 
$\gamma^-_p$ is the length of chord with endpoints on $\mathcal C_p$, $p \in (1,2]$, defining an arc of length $e^-_p$.
Similarly, $\gamma^+_p$ is the length of a chord with endpoints on $\mathcal C_p$, $p \in (2,\infty)$, defining an arc of length $e^+_p$.
Unlike the Euclidean unit disks, in $\ell_p$ unit disks, arc and chord lengths are not invariant under rotation. In other words , arbitrary chords of length $\gamma^-_p, \gamma^+_p$ do not necessarily correspond to arcs of length $e^-_p$, and $e^+_p$, respectively ,and vice versa.
The claim extends also to the $\ell_1, \ell_\infty$ spaces. For a simple example, consider points 
$
A=\rho_1(\pi/4)=(1/2,1/2), 
B=\rho_1(3\pi/4)=(-1/2,1/2), 
C=\rho_1(0)=(1,0), 
D=\rho_1(\pi/2)=(0,1)$.
It is easy to see that 
$d_p(A,B)=1$ and $d_p(C,D)=2$, while 
$\arccc d_p (A,B) = \arccc d_p (C,D)= \pi_1/2=2$, in other words two arcs of the same length identify chords of different length. 
We are motivated to introduce the following definition. 

\begin{definition}
\label{def: min chord for arc}
For every $p \in [1,\infty)$, and for every $u \in [0,2\pi_p)$, we define 
$$
\mathcal L_p(u) := \min_{A,B\in \mathcal C_p}\left\{ 
\norm{A-B}_p:  
~~\mu_p\left( \arccc{AB} \right)=u
\right\}.
$$
\end{definition}
In other words $\mathcal L_p(u)$ is the length of the shortest line segment (chord) with endpoints in $\mathcal C_p$ at arc distance $u$ (and corresponding to an arc of measure $u$), and hence $\mathcal L_p (u) = \mathcal L_p (2\pi_p -u)$ for every $u \in (0,2\pi_p)$. 
 As a special example, note that $\mathcal L_2(u)= 2\sin(u/2)$, as well as $\mathcal L_p(\pi_p)=2$, for all $p\in [1,\infty)$. 

\begin{lemma}
\label{lem: mathcal L monotonicity}
For every $p \in (1,\infty)$, function $\mathcal L_p(u)$ is increasing in $u \in [0,\pi_p]$.
\end{lemma}

The intuition behind Lemma~\ref{lem: mathcal L monotonicity} is summarized in the following proof sketch. Assuming, for the sake of contradiction, that the lemma is false, there must exist an interval of arc lengths, and some $p\geq 1$ for which $\mathcal L_p(u)$ is strictly decreasing. By first-order continuity of $\norm{A-B}_p$, and in the same interval of arc-lengths, chord $\norm{A-B}_p$ must be decreasing in $\mu_p\left( \arccc{AB} \right)$ even when points $A,B$ are conditioned to define a line with a fixed slope (instead of admitting a slope that minimizes the chord length). However, the last statement gives a contradiction (note that the chord length is maximized at the diameter, when the corresponding arc has length $\pi_p$). Indeed, consider points $A,B,A',B'$ such that $\arccc d_p (A,B)=u, \arccc d_p (A',B')=u'$, with $u<u'\leq \pi_p$. It should be intuitive that $d_p(A,B) \leq d_p (A',B) \leq d_p (A',B')$
(note that the proof of Lemma~\ref{equa: distance robots} shows the monotonicity of $\delta_{p,0}(\tau), \delta_{p,\pi/4}(\tau)$, which is used in the previous statement when $AB,A'B'$ are either parallel to the line $x=0$ or parallel to the line $y=-x$).

For fixed values of $p$, Lemma~\ref{lem: mathcal L monotonicity} can also be verified with confidence of at least 6 significant digits in MATHEMATICA. 
Due to precision limitations, the values of $p$ cannot be too small, neither too big, even though a modified working precision can handle more values of $p$. 
With standard working precision, any $p$ in the range between $1.001$ and $45$ can be handled within a few seconds. 
As we argue later, for large values of $p$, Lemma~\ref{lem: mathcal L monotonicity} bears less significance, since in that case we have an alternative way to prove the (near) optimality of algorithms Wireless-Search$_p$($\phi$), as per Theorem~\ref{thm: upper bound best algo}.
Next we provide a visual analysis of function $\mathcal L_p$ that effectively justifies Lemma~\ref{lem: mathcal L monotonicity}, see Figure~\ref{fig: LambdaP3Dx2}.

\begin{figure}[h!]
  \centering
  \includegraphics[width=0.9\linewidth]{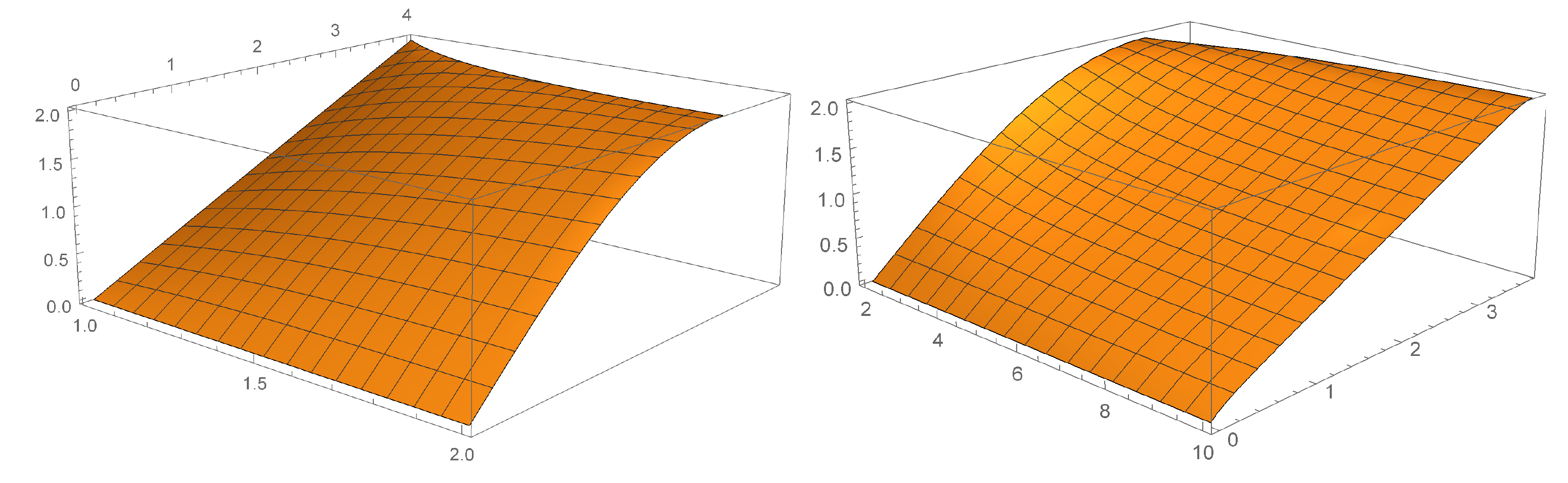}
\caption{
Figures depict $\mathcal L_p(u)$ for various values of $p$ and for $u \in [0,\pi_p]$. 
Left-hand side figure shows increasing function $\mathcal L_p(u)$ for $p\in (1,2]$. 
Right-hand side figure shows increasing function $\mathcal L_p(u)$ for $p\in [2,10]$. 
Recall that $\mathcal L_2(u)= 2\sin(u/2)$, $\mathcal L_p(\pi_p)=2$, for all $p\in [1,\infty)$, as well as that $\pi_1=\pi_\infty=4$ and $\pi_p<4$ for $p\in (1,\infty)$. 
}
\label{fig: LambdaP3Dx2}
\end{figure}

\section{Visualization of Key Concepts and Results}
\label{sec: visualizations}

In this section we provide visualizations of some key concepts, along with visualizations of our results. The Figures are referenced in various places in our manuscript but we provide self-contained descriptions.

\begin{figure}[h!]
\begin{subfigure}[t]{0.45\textwidth}
\centering
\includegraphics[width =2.3in]{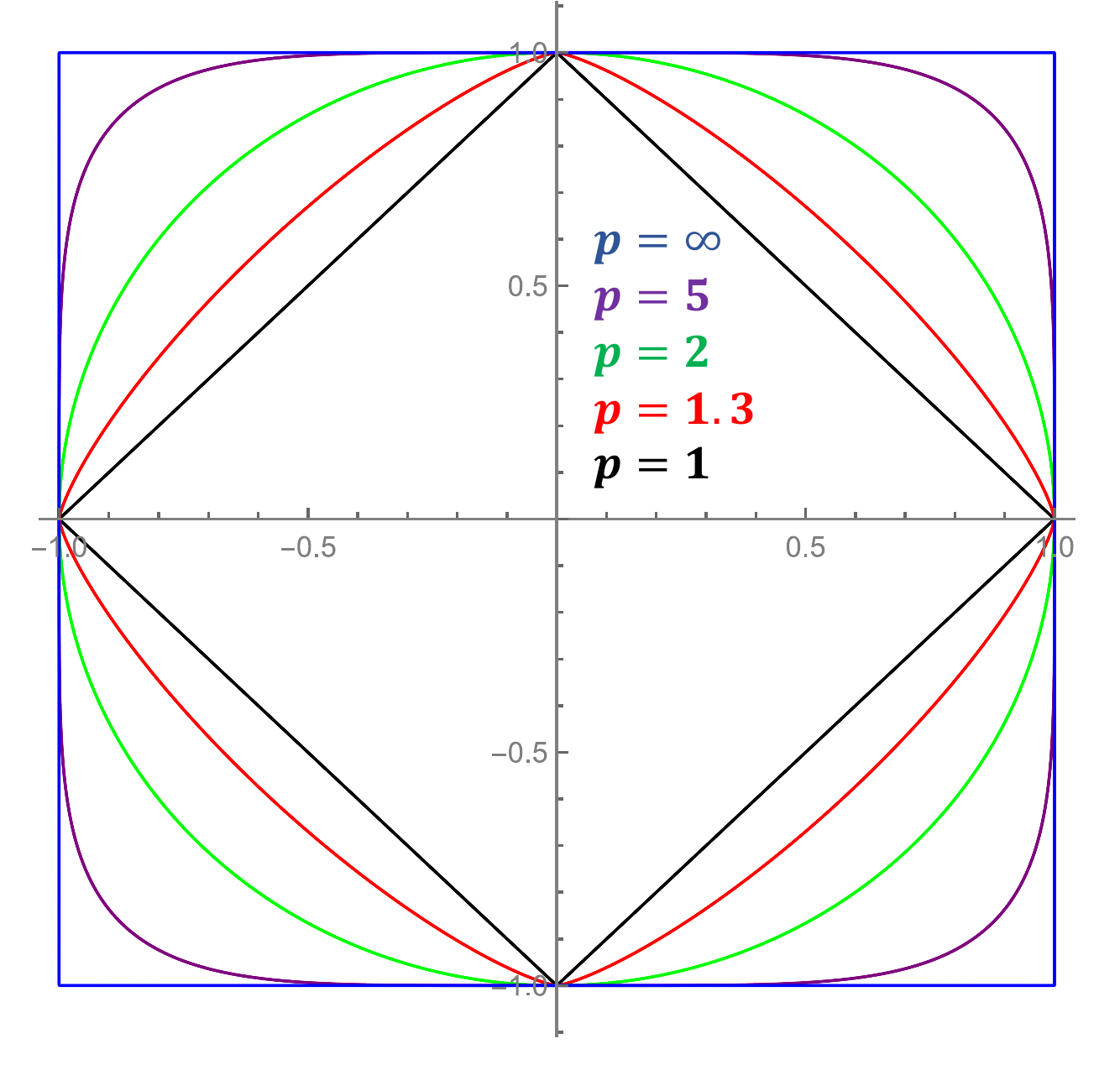}
              \caption{Unit circles $\mathcal C_p$, for $p=1,1.3,2,5,\infty$, induced by the $\ell_p$ norm.}
              \label{fig: UnitCircles}
\end{subfigure}\hfill
\begin{subfigure}[t]{0.5\textwidth }
\centering
\includegraphics[width =2.9in]{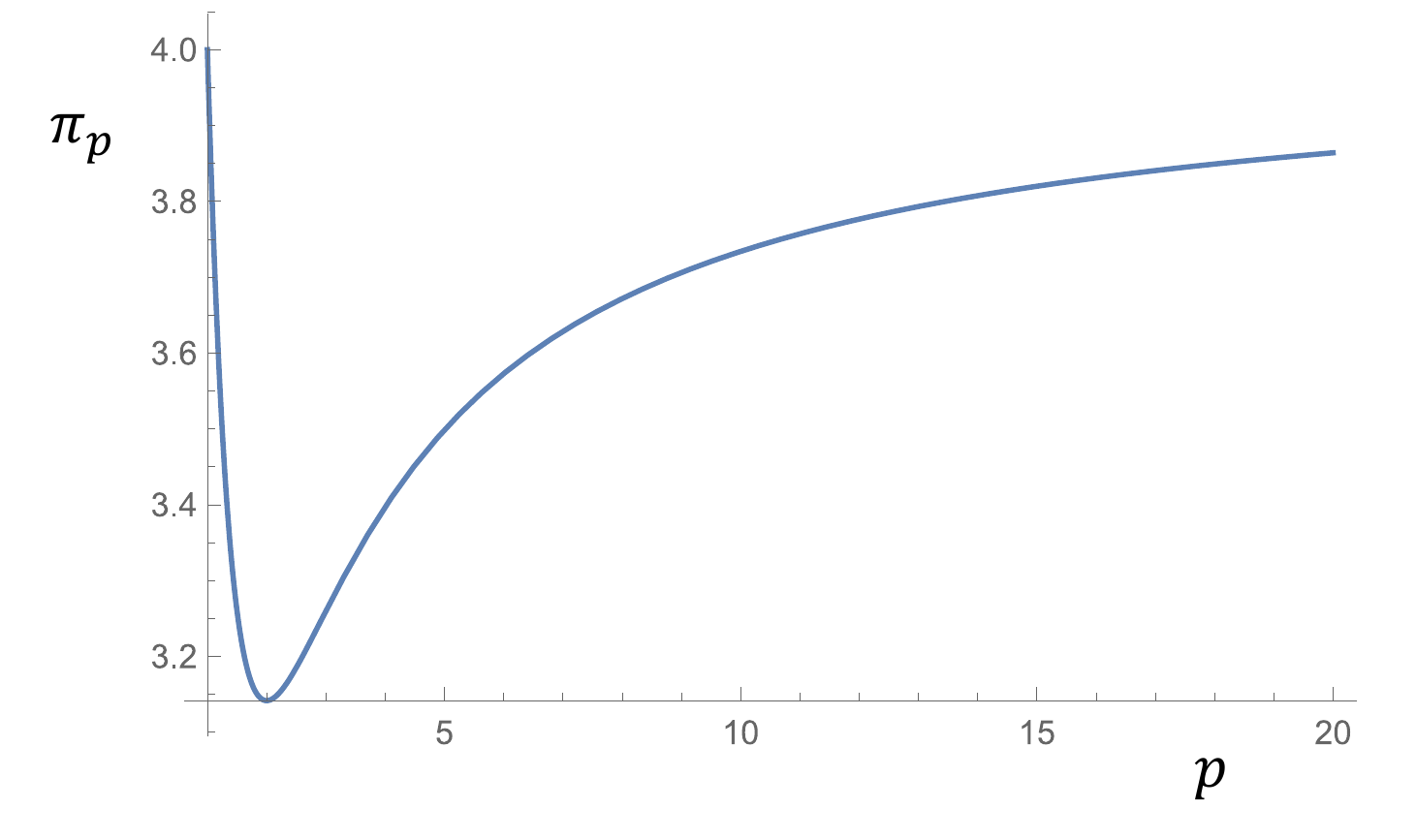}
              \caption{The behavior of $\pi_p$ as $p$ ranges from 1 up to $\infty$, where $\pi_1=\pi_\infty=4$ and $\pi_2=\pi$ is the smallest value of $\pi_p$.}
              \label{fig: ValueOfPi}
\end{subfigure}
\caption{}
\end{figure}

Figures~\ref{fig: GammapEp} and~\ref{fig: RelativeEp} depict quantities pertaining to algorithm Wireless-Search$_p$($\phi$) (where $\phi=0$, if $p\in [1,2)$ and $\phi=\pi/4$, if $p\in (2,\infty)$) for the placement of the hidden exit inducing the worst case cost. 
Moreover Figure~\ref{fig: GammapEp}
depicts quantities $e_p/2, \gamma_p$, as per Theorem~\ref{thm: explored optimal algo}. In particular, for each $p$, quantity $e_p/2$ is the time a searcher has spent searching the perimeter of $\mathcal C_p$ till the hidden exit is found (in the worst case).
Therefore, $e_p$ represents the portion of the perimeter that has been explored till the exit is found.               
              Quantity $\gamma_p$ is the distance of the two robots at the moment the hidden exit is found so that the total cost of the algorithm is $1+e_p/2+\gamma_p$. 
              By~\cite{CGGKMP} we know that $e_2=4\pi/3$ and $\gamma_2=\sqrt3$. 
              Our numerical calculations also indicate that 
$\lim_{p\rightarrow 1} e_p=12/5$,
$\lim_{p\rightarrow 1} \gamma_p=8/5$, and
$\lim_{p\rightarrow \infty} e_p=\lim_{p\rightarrow \infty} \gamma_p=2$.

              Figure~\ref{fig: RelativeEp} depicts quantities $e_p/2\pi_p$, which equals the explored portion of the unit circle $\mathcal C_p$, relative to its circumference, of Algorithm Wireless-Search$_p$($\phi$) (where $\phi=0$, if $p\in [1,2)$ and $\phi=\pi/4$, if $p\in (2,\infty)$)
              when the worst case cost inducing exit is found. 
                By~\cite{CGGKMP} we know that $e_2/2\pi_2=(4\pi/3)/2\pi=2/3$. 
                Interestingly, quantity $e_p/2\pi_p$ is maximized when $p=2$, that is in the Euclidean plane searchers explore the majority of the circle before the exit is found, for the placement of the exit inducing worst case cost. 
Also, numerically we obtain that
$\lim_{p\rightarrow 1} e_p/2\pi_p=3/5$, and
$\lim_{p\rightarrow \infty} e_p/2\pi_p=1/2$. 
The reader should contrast the limit valuations with Lemma~\ref{lem: upper bound 1 and infty} according which in both cases $p=1,\infty$ the cost of our search algorithms is constant and equal to 5 for all placements of the exit that are found from the moment searchers have explored half the unit circle and till the entire circle is explored.

\begin{figure}[h!]
\begin{subfigure}[t]{0.47\textwidth}
\centering
\includegraphics[width =3in]{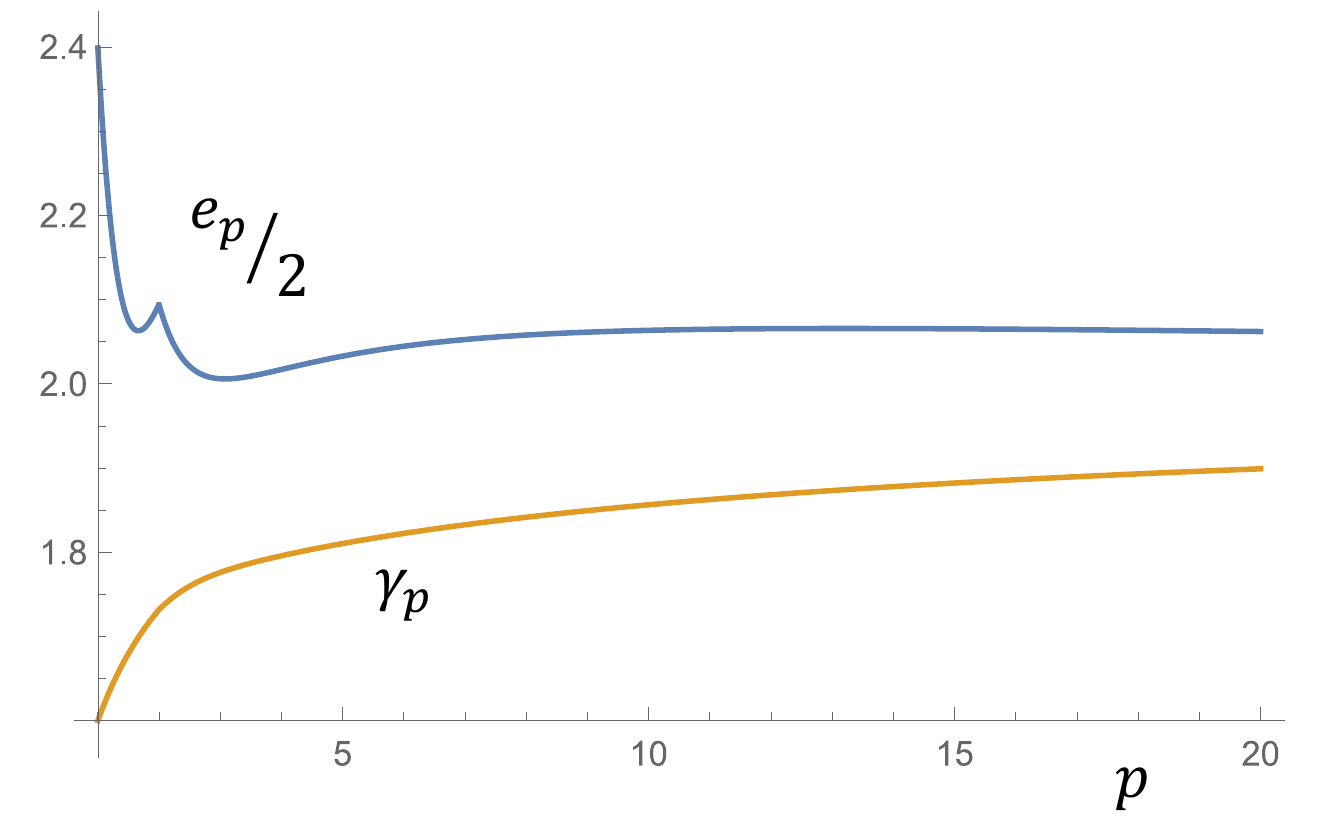}
              \caption{
Perimeter search time $e_p/2$ and distance $\gamma_p$ between searchers when worst case cost inducing exit is found as a function of $p$, see also Theorem~\ref{thm: explored optimal algo}. 
              }
              \label{fig: GammapEp}
\end{subfigure}\hfill
\begin{subfigure}[t]{0.47\textwidth }
\centering
\includegraphics[width =2.9in]{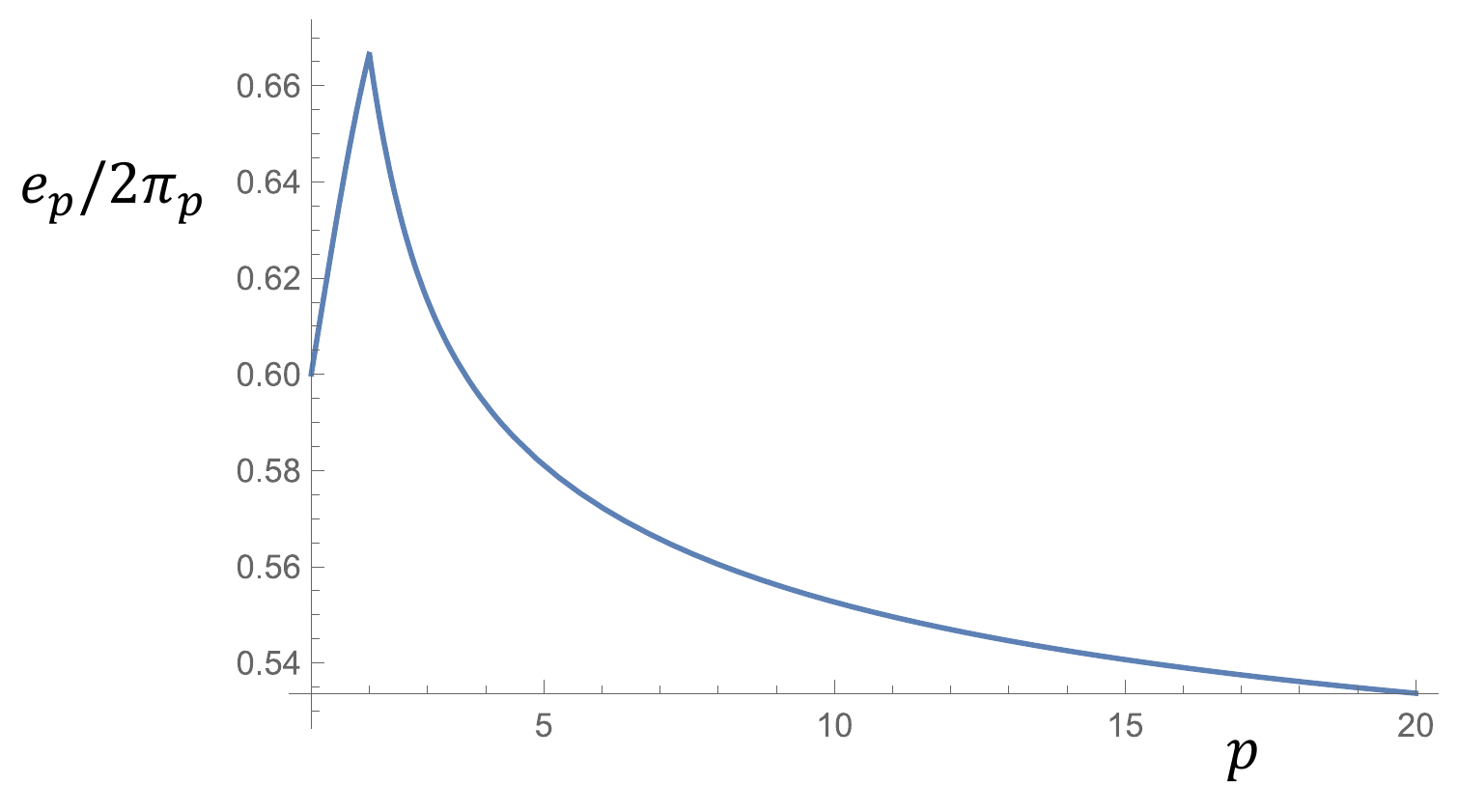}
              \caption{
Explored portion $e_p/2\pi_p$ as a function of $p$. 
}
              \label{fig: RelativeEp}
\end{subfigure}
\caption{}
\end{figure}

Figure~\ref{fig: AlgoPerformance} shows the worst case performance analysis of Algorithm Wireless-Search$_p$($\phi$) (where $\phi=0$, if $p\in [1,2)$), which is also optimal for problem \textsc{WE}$_p$.
As per Lemma~\ref{lem: upper bound 1 and infty}, the evacuation cost is 5 for $p=1$ and $p=\infty$. 
The smallest (worst case) evacuation cost when $p\in [1,2]$ is $4.7544$ and is attained at $p\approx 1.5328$.
The smallest (worst case) evacuation cost when $p\in [2,\infty]$ is $4.7784$ and is attained at $p\approx 2.6930$.
As per~\cite{CGGKMP}, the cost is $1+\sqrt3+2\pi/3 \approx 4.82644$ for the Euclidean case $p=2$.

\begin{figure}[h!]
\begin{subfigure}[t]{0.43\textwidth}
\centering
\includegraphics[width =2.1in]{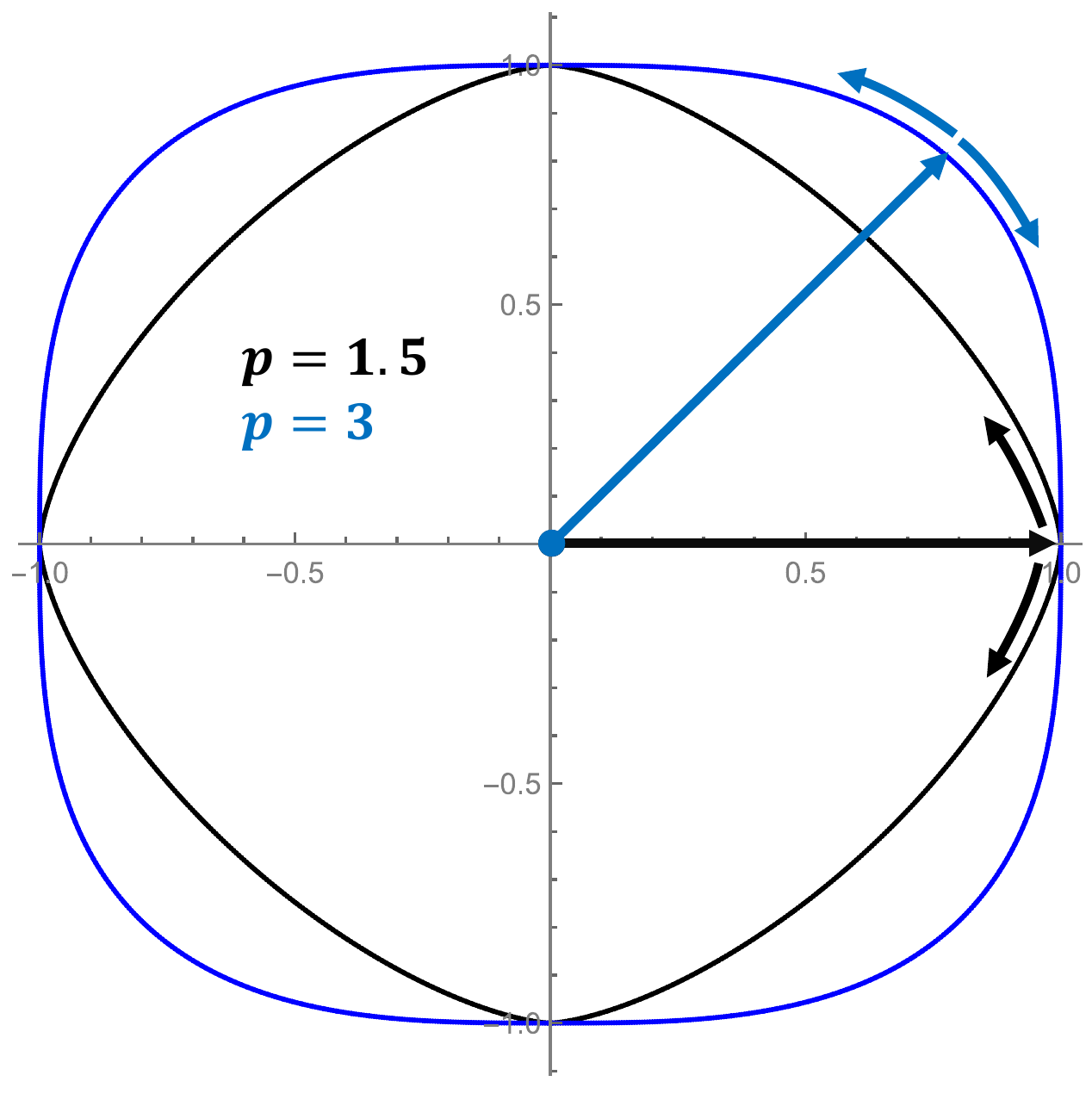}
              \caption{
              Figure depicts robots' trajectories for 
              algorithms Wireless-Search$_{1.5}$(0) and Wireless-Search$_3$($\pi/4$).
              }
              \label{fig: SearchingTwoCircles}
\end{subfigure}\hfill
\begin{subfigure}[t]{0.52\textwidth }
\centering
\includegraphics[width =2.6in]{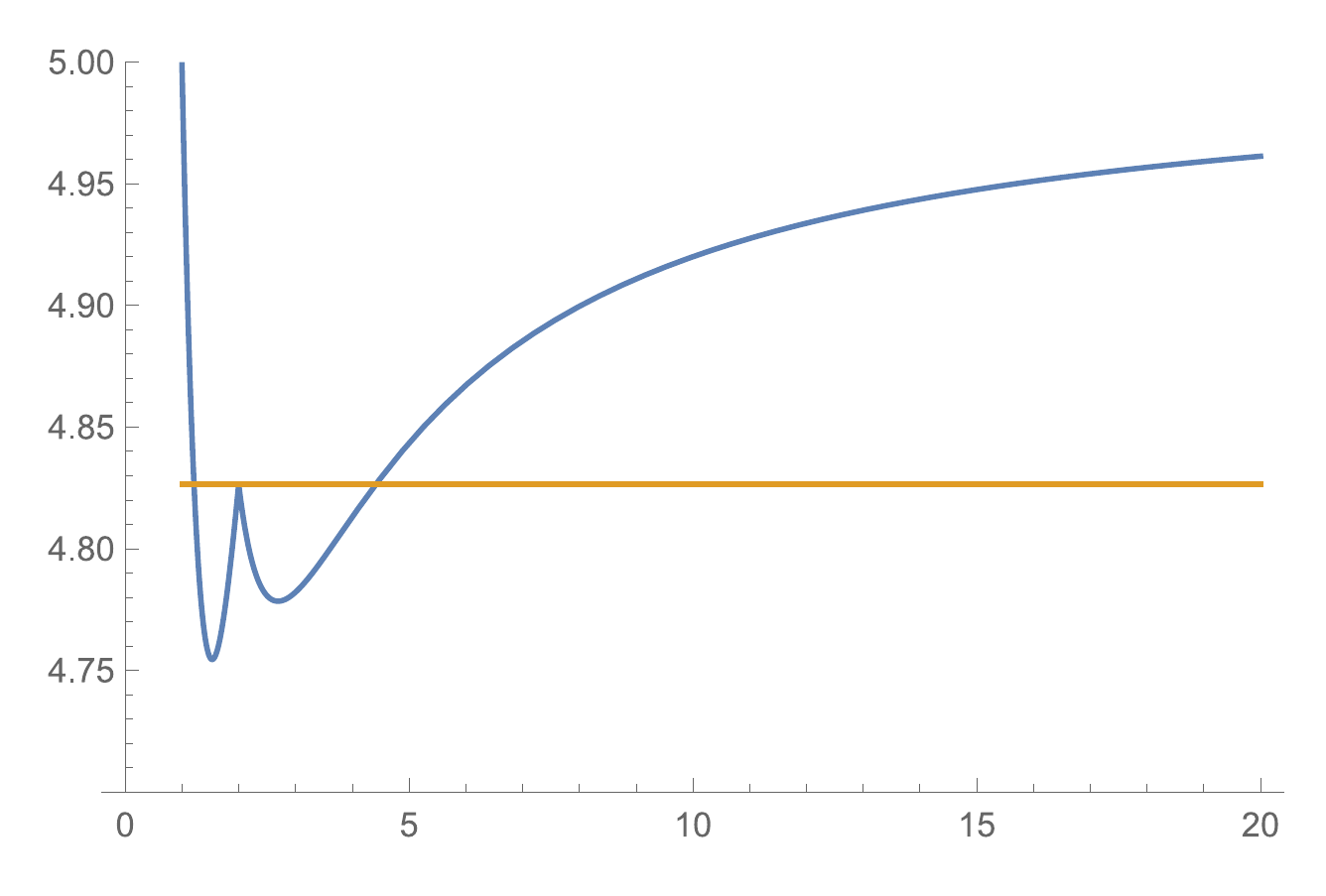}
              \caption{
Blue curve depicts the worst case evacuation cost of Algorithm Wireless-Search$_p$($\phi$), where $\phi=0$, if $p\in [1,2)$, as a function of $p$ . 
Yellow line is the optimal evacuation cost in the Euclidean metric space. 
              }
              \label{fig: AlgoPerformance}
\end{subfigure}
\caption{}
\end{figure}

\section{Lower Bounds \& the Proof of Theorem~\ref{thm: upper bound best algo}}
\label{sec: lower bounds}

First we prove a weak lower bound that holds for all $\ell_p$ spaces, $p\geq 1$ (see also Figure~\ref{fig: ValueOfPi} for a visualization of $\pi_p$). 

\begin{lemma}
\label{sec: lower bound weak}
For every $p \in [1,\infty)$, the optimal evacuation cost of \textsc{WE}$_p$ is at least $1+\pi_p$. 
\end{lemma}

\begin{proof}
The circumference of $\mathcal C_p$ has length $2\pi_p$. 
Two unit speed robots can reach the perimeter of $\mathcal C_p$ in time at least 1. 
Since they are searching in parallel, in additional time $\pi_p-\epsilon$, they can only search at most $2\pi_p-2\epsilon$ measure of the circumference. 
Hence, there exists an unexplored point $P$. 
Placing the exit at $P$ shows that the evacuation time is at least $1+\pi_p - 2\epsilon$, for every $\epsilon>0$. 
\qed \end{proof}

In particular, recall that $\pi_1=\pi_\infty=4$, and hence no evacuation algorithm for \textsc{WE}$_1$ and \textsc{WE}$_\infty$  has cost  less than 5. 
As a corollary, we obtain that 
Algorithms Wireless-Search$_1$(0) and Wireless-Search$_\infty$($\pi/4$) are optimal, hence proving the special cases $p=1,\infty$ of Theorem~\ref{thm: upper bound best algo}. The remaining cases require a highly technical treatment. 

The following is a generalization of a result first proved in~\cite{CGGKMP} for the Euclidean metric space (see Lemma 5 in the Appendix of the corresponding conference version). The more general proof is very similar. 

\begin{lemma}
\label{lem: arc chord points}
For every $V \subseteq \mathcal C_p$, with $\mu_p(V) \in (0,\pi_p]$, and for every small $\epsilon >0$, there exist $A,B \in V$ with $\arccc d_p (A,B) \geq \mu_p(V) - \epsilon$. 
\end{lemma}

\begin{proof}

For the sake of contradiction, consider some $V \subseteq \mathcal C_p$, with $\mu_p(V) \in (0,\pi_p]$, and some small $\epsilon >0$, such no two points both in $V$ have arc distance at least $\mu_p(V) - \epsilon$. Below we denote the latter quantity by $u$, and note that $u \in (0, \pi_p)$, as well as that $\mu_p(V)=u+\epsilon>u$. We also denote by $V^\complement$ the set $\mathcal C_p \setminus V$. 
The argument below is complemented by Figure~\ref{fig: ArcChordPoints}.
\begin{figure}[t]
  \centering
  \includegraphics[width=0.6\linewidth]{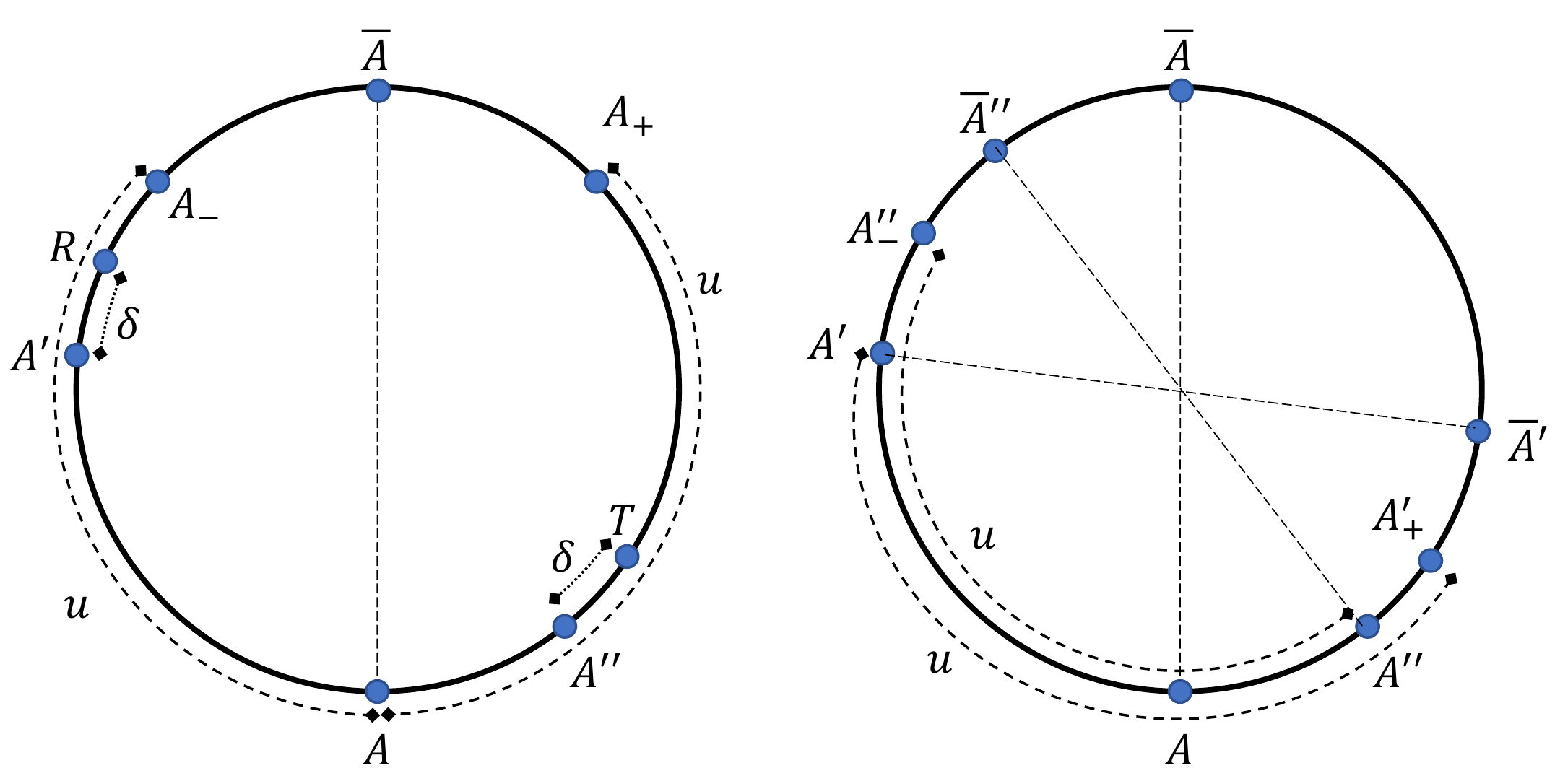}
\caption{
An abstract $\ell_p$ unit circle, for some $p\geq 1$, depicted as a Euclidean unit circle for simplicity. 
On the left we depict points $A, \overline A, A_-,A_+,A',A'',T,R$. On the right we keep only the points relevant to our final argument, and add also points $A_-'',A_+', \overline{A}', \overline{A}''$. }
\label{fig: ArcChordPoints}
\end{figure}

Since $V$ is non-empty, we consider some arbitrary $A \in V$. We define the set of \emph{antipodal} points of $V$
$$
N := \{ B \in \mathcal C_p: \exists C \in V, \arccc d_p (B,C) = \pi_p \}
$$
Note that $N\cap V = \emptyset$ as otherwise we have a contradiction, i.e. two points in $V$ with arc distance $\pi_p > u= \mu_p(V) - \epsilon$. In particular, we conclude that $N \subseteq V^\complement$, and hence by Lemma~\ref{lem: symmetries} we have $\mu_p(N) = \mu_p(V)=u+\epsilon$. 

Let $\barrrr{A}$ be the point antipodal to $A$, i.e. $\arccc d_p(A,\barrrr A) = \pi_p$. 
Next, consider points $A_-, A_+ \in \mathcal C_p$ at anti-clockwise and clockwise arc distance $u$ from $A$, that is $\arccc d_p (A,A_-) =\arccc d_p (A_+,A)=u$. All points in $\arccc{A_+A_-}$ are by definition at arc distance at least $u$ from $A$.
In particular, $\barrrr A \in \arccc{A_+A_-}$ and $A_- \in \arccc{\barrrr{A}A}, A_+ \in \arccc{A\barrrr{A}}$. 
 We conclude that $V \cap \arccc{A_+A_-} = \emptyset$, as otherwise we have $A\in V$ together with some point in $V \cap \arccc{A_+A_-}$ make two points with arc distance at least $u$. Note that this implies also that $\arccc{A_+A_-} \subseteq V^\complement$. 

Consider now the minimal, inclusion-wise, arc $\arccc{TR} \subseteq V^\complement$, containing $\arccc{A_+A_-}$. Such arc exists because $A_-,A_+ \in \arccc{A_+A_-} \subseteq V^\complement$. In particular, since $A\in V$, we have that $R \in \arccc{A_-A}$ and $T \in \arccc{AA_+}$. 

For some arbitrarily small $\delta >0$, with $\delta < \min\{ u, \epsilon/2\}$, let $A',A''\in V$ such that $\arccc d( RA') = d(A''T)=\delta$. Such points $A',A''$ exist, as otherwise $\arccc{TR}$ would not be minimal. 
Clearly, we have $A' \in \arccc{RA}$ and $A'' \in \arccc{AT}$. 

Since $A' \in \arccc{RA} \subseteq \arccc{\barrrr{A}A}$, its antipodal point $\barrrr{A}'$ lies in $\arccc{A\barrrr{A}}$. 
Similarly, since $A'' \in \arccc{AT} \subseteq \arccc{A\barrrr{A}}$, its antipodal point $\barrrr{A}''$ lies in $\arccc{\barrrr{A}A}$. 
Finally, we consider point $A_-''$ at clockwise arc distance $u$ from $A''$, and point $A_+'$ at anti-clockwise distance $u$ from $A'$, 
that is $\arccc{d}(A_-'',A'') = \arccc{d}(A',A_+') = u$. 
We observe that $A_-'' \in \arccc{\barrrr{A}''A}$ and $A_+' \in \arccc{A\barrrr{A}'}$.

Recall that $A''\in V$, hence $\arccc{\barrrr{A}''A_-''} \subseteq V^c$, as otherwise any point in $\arccc{\barrrr{A}''A_-''} \cap V$ together with $A'$, at arc distance at least $u$, would give a contradiction. 
Similarly, since $A'\in V$, hence $\arccc{A_+'\barrrr{A}'} \subseteq V^c$, as otherwise any point in $\arccc{A_+'\barrrr{A}'} \cap V$ together with $A''$, at arc distance at least $u$, would give a contradiction.

Lastly, abbreviate $\arccc{\barrrr{A}''A}, \arccc{A_+'\barrrr{A}'}$ by $X,Y$, respectively. 
Note that 
$
\mu_p(X) 
= \mu_p (\arccc{\barrrr{A}''A''} \setminus \arccc{A_-''A''})  
= \mu_p (\arccc{\barrrr{A}''A''}) - \mu_p(\arccc{A_-''A''})
= \pi_p - u.
$
Similarly, 
$
\mu_p(Y) 
= \mu_p (\arccc{A'\barrrr{A}'} \setminus \arccc{A'A_+'})  
= \mu_p (\arccc{A'\barrrr{A}'}) - \mu_p(\arccc{A'A_+'})
= \pi_p - u.
$
Recall that $A_-'' \in \arccc{\barrrr{A}''A}$ and $A_+' \in \arccc{A\barrrr{A}'}$, and hence sets $X, Y$ intersect either at point $A$ or have empty intersection. As a result
$
\mu_p \left( 
X \cap Y 
\right) =0,
$
as well as 
$
\mu_p \left( 
N \cap X \cap Y 
\right)
=0
$.

Recall that $\mu_p(\arccc{RA'}) = \mu_p(\arccc{A''T}) =  \delta$, and so by Lemma~\ref{lem: symmetries} we also have 
$
\mu_p \left( 
X \cap N 
\right)
=
\mu_p \left( 
Y \cap N 
\right) = \delta.
$
But then, using inclusion-exclusion for measure $\mu_p$, we have
\begin{align*}
\mu_p( N \cup X \cup Y ) 
&=
\mu_p(N) + \mu_p(X) + \mu_p(Y) 
- \mu_p(N\cap X) - \mu_p(N\cap Y)- \mu_p(X\cap Y)
+ \mu_p(N\cap X \cap Y) \\
& = u+\epsilon + \pi - u + \pi - u  - \delta - \delta - 0 + 0 \\
& = 2\pi_p - u + \epsilon - 2\delta \\
& > 2\pi_p-u \\
& > 2\pi_p-\mu_p(V) \\
& = \mu_p(V^\complement).
\end{align*}
Hence $\mu_p( N \cup X \cup Y )  >  \mu_p(V^\complement)$.
On the other hand, recall that $N,X,Y \subseteq V^\complement$, hence $N \cup X \cup Y  \subseteq V^\complement$, hence $\mu_p( N \cup X \cup Y )  \leq  \mu_p(V^\complement)$, which is a contradiction. 
\qed \end{proof}

We are now ready to prove a general lower bound for \textsc{WE}$_p$ which we further quantify later. 

\begin{lemma}
\label{lem: lower bound generic}
For every $p \in (1,\infty)$, the optimal evacuation cost of \textsc{WE}$_p$ is at least $1+e_p/2 + \mathcal L_p(e_p)$. 
\end{lemma}

\begin{proof}
Consider an arbitrary evacuation algorithm $\mathcal A$. We show that the cost of $\mathcal A$ is at least $1+e_p + \mathcal L_p(e_p)$. 
By Theorem~\ref{thm: explored optimal algo}, we have that $e_p \in (\pi_p, 2\pi_p]$. Let $\epsilon>0$ be small enough, were in particular $\epsilon <e_p-\pi_p$. We let evacuation algorithm $\mathcal A$ run till robots have explored exactly $e_p-\epsilon$ part of $\mathcal C_p$. 

The two unit speed robots need time 1 to reach the perimeter of $\mathcal C_p$.
Since moreover they (can) search in parallel (possibly different parts of the unit circle), they need an additional time at least $(e_p-\epsilon)/2$ in order to explore measure $e_p-\epsilon$. 
The unexplored portion $V$ of $\mathcal C_p$ has therefore measure $u:=2\pi_p - (e_p-\epsilon)$, where $u \in (0,\pi_p)$. 

By Lemma~\ref{lem: arc chord points}, there are two points $A,B \in V$ that are at an arc distance $v\geq u-\epsilon = 2\pi_p - e_p$. 
By definition, both points $A,B$ are unexplored. 
We let algorithm $\mathcal A$ run  even more and till the moment any one of the points $A,B$ is visited by some robot, and we place the exit at the other point (even if points are visited simultaneously), hence algorithm $\mathcal A$ needs an additional time $d_p(A,B)$ to terminate, for a total cost at least $1+e_p/2-\epsilon/2+d_p(A,B)$. But then, note that 
$
d_p(A,B) \geq \mathcal L_p(v)\geq \mathcal L_p(2\pi_p - e_p),
$
where the first inequality is due Definition~\ref{def: min chord for arc}
and the second inequality due to Lemma~\ref{lem: mathcal L monotonicity}, and the claim follows by recalling that $\mathcal L_p(2\pi_p-e_p)=\mathcal L_p(e_p)$. 
\qed \end{proof}

Recall that, for every $p\in (1, \infty)$,  the evacuation algorithms we have provided for \textsc{WE}$_p$ have cost $1+e_p/2+\gamma_p$.\footnote{This is unless, by Lemma~\ref{lem: performance phi=pi_p/4} and for $p\geq 2$, we have that $E_{p,\pi/4}=1+\pi_p$. However, in the latter case we can invoke Lemma~\ref{sec: lower bound weak} according to which Algorithm Wireless-Search$_p$($\pi/4$) would still be optimal. Hence, we may assume w.l.o.g that  
$E_{p,\pi/4}\not =1+\pi_p$ and that $E_{p,\pi/4}$ is given by the alternative formula of Lemma~\ref{lem: performance phi=pi_p/4}.}
At the same time, Lemma~\ref{lem: lower bound generic} implies that no evacuation algorithm has cost less than $1+e_p/2 + \mathcal L_p(e_p)$. So, the optimality of our algorithms, 
that is, the proof of Theorem~\ref{thm: upper bound best algo}, is implied directly by the following lemma, which is verified numerically.
The details are presented in the next section. 

\begin{lemma}
\label{lem: min arc is what I used}
For every $p \in (1,\infty)$, we have $\mathcal L_p(e_p)=\gamma_p$. 
\end{lemma}

\section{Numerical Verification of Lemma~\ref{lem: min arc is what I used}}
\label{sec: critical arc orientation}

Consider a contiguous arc of $\mathcal C_p$ of length $e_p$. As the endpoints of the arc move around the perimeter of $\mathcal C_p$, the length of the corresponding chord, i.e. line segment with the same endpoints, changes. 
Lemma~\ref{lem: min arc is what I used} states that the shortest such length is equal to $\gamma_p$, as per Theorem~\ref{thm: explored optimal algo}. 

For an arbitrary contiguous arc $\arccc{AB}$ of $\mathcal C_p$ of length $e_p$, let $M$ be the midpoint of the arc, i.e. point $M$ satisfies 
$\mu_p\left( \arccc{BM} \right)=\mu_p\left( \arccc{MA} \right)=e_p/2$.
We define the \emph{tangential angle} of the arc $\arccc{AB}$ as the angle $\theta$ satisfying $\rho_p(\theta)=M$. 
In other words, the tangential angle of an arc assumes values in $[0,2\pi)$. 

Clearly, as the tangential angle of a fixed-length arc varies in $[0,2\pi)$, the length of the corresponding chord changes. 
At the same time, by the symmetries of $\mathcal C_p$, all possible chord length values are attained as the tangential angle ranges in $[0,\pi/4]$.
In other words, Lemma~\ref{lem: min arc is what I used} states that as the tangential angle of a contiguous arc of length $e_p$ ranges in $[0,\pi/4]$, the minimum length of the corresponding chord equals $\gamma_p$. 

It is now informative to recall the definition of $\gamma_p$, which is the $\ell_p$ distance of the two searchers at the moment the exit is found and for the placement of the exit that induces the worst case cost of algorithm Wireless-Search$_p$($\phi$), where $\phi=0$, when $p\in [1,2)$ and $\phi=\pi/4$ when $p\in (2,\infty)$. 
In particular (see Lemma~\ref{lem: sym trajectories}), when $p \in [1,2)$ the positions of the robots define arcs with tangential angle 0, whereas when $p \in (2,\infty)$ the positions of the robots define arcs with tangential angle $\pi/4$. 
Stated differently, $\gamma_p$ is by definition, the length of a chord corresponding to an arc of length $e_p$ that has tangential angle $0$ if $p\in [1,2)$ and $\pi/4$ when $p \in (2,\infty)$.
In order to formally restate Lemma~\ref{lem: min arc is what I used}, we  introduce the following notation. Function $\sigma_p:[0,\pi/4]\mapsto \reals$ is defined as the length of the chord, corresponding to an arc of length $e_p$ with tangential angle $\theta$. 
Using this notation, we need to show that $\min_{\theta \in [0,\pi/4]} \sigma_p(\theta)=\gamma_p$, which is exactly what the next lemma states. 


\begin{lemma}
\label{lem: min arc is what I used restate}
Function $\sigma_p$ is minimized at $\theta=0$ when $p\in [1,2)$ and at $\theta =\pi/4$ when $p \in (2,\infty)$. 
\end{lemma}

For fixed values of $p$, Lemma~\ref{lem: min arc is what I used restate} can be verified with confidence of at least 6 significant digits in MATHEMATICA. 
Due to precision limitations, the values of $p$ cannot be too small, neither too big, even though a modified working precision can handle more values. 
With standard working precision, any $p$ in the range between $1.001$ and $20$ can be handled within a few seconds. 
For large values of $p$, Lemma~\ref{sec: lower bound weak} gives a nearly tight bound. For example, if $p=1,000$, the performance of our algorithm is $4.9993023351$, while the lower bound of Lemma~\ref{sec: lower bound weak} is $1+\pi_{1000}\approx 4.9972283728$.

\ignore{
standard precision values of $p$ cannot be arbitrarily close to 1, but $p=1.001$ can be handled without modifications. Similarly, $p$ cannot be arbitrarily large, but the default precision can handle values up to $p=45$. 
By increasing the machine precision, one can compute $\sigma_p$ and the performance of algorithm Wireless-Search$_p$($\pi/4$) for much higher values. Nevertheless, for large values, Lemma~\ref{sec: lower bound weak} gives a nearly tight bound. For example, if $p=1,000$, the performance of our algorithm is $4.9993023351$, while the lower bound of Lemma~\ref{sec: lower bound weak} is $1+\pi_{1000}\approx 4.9972283728$.
}

Next we provide a visual analysis of function $\sigma_p$ that effectively justifies Lemma~\ref{lem: min arc is what I used restate}. In fact, we show the following stronger statement, see Figures~\ref{fig: SigmpaP3Dx2},\ref{fig: SigmpaP2Dx6small},\ref{fig: SigmpaP2Dx6large}.

\begin{lemma}
\label{lem: min arc is what I used restate x2}
Function $\sigma_p:[0,\pi/4]\mapsto \reals$ is increasing when $p\in [1,2)$ and decreasing when $p \in (2,\infty)$.
\end{lemma} 

Note that function $\sigma_2$ is constant. In particular, its value equals the distance of the robots, in the worst placement of the exit, the moment the exit is found, when searching in the Euclidean space. Since $e_2=4\pi/3$, it follows that $\sigma_2(\theta)=\gamma_2=\sqrt3$, for all $\theta\in [0,\pi/4]$. 

\begin{figure}[h!]
  \centering
  \includegraphics[width=0.9\linewidth]{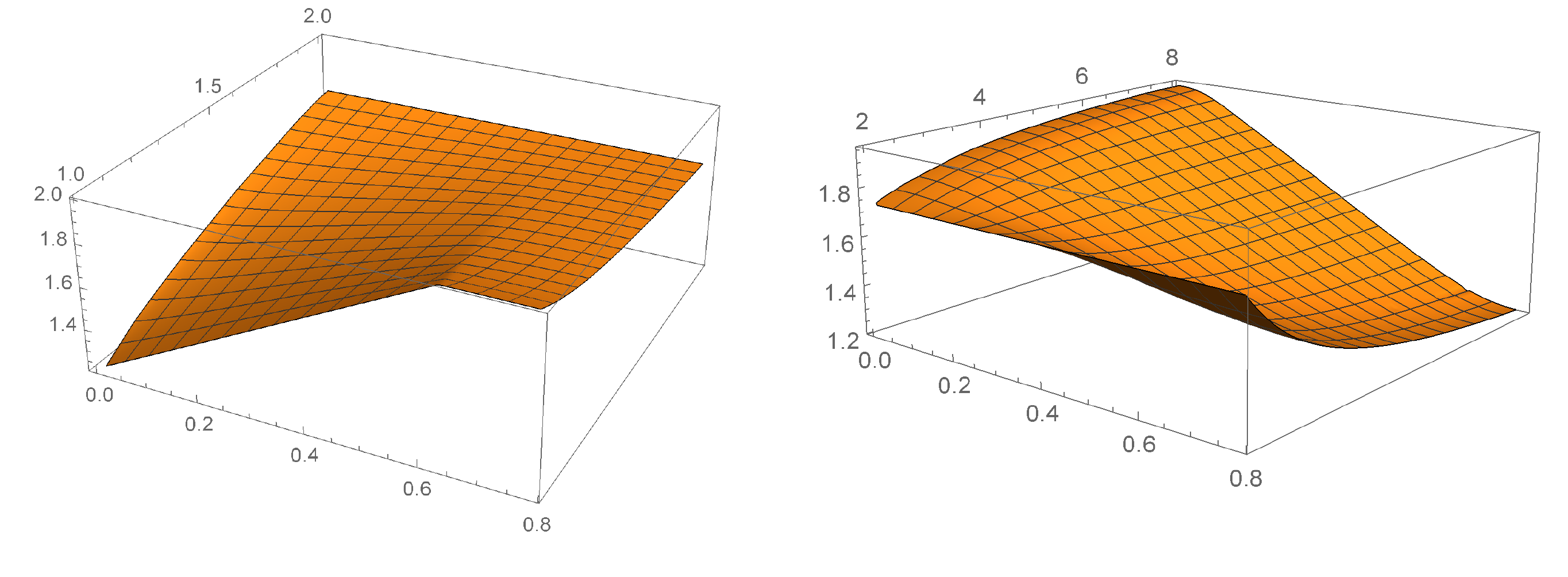}
\caption{
Figures depict $\sigma_p(\theta)$ for various values of $p$ and for $\theta \in [0,\pi/4]$. Left-hand side figure shows increasing function $\sigma_p$ for $p\in (1,2)$. 
Right-hand side figure shows decreasing function $\sigma_p$ for $p\in (2,10]$. 
}
\label{fig: SigmpaP3Dx2}
\end{figure}

\begin{figure}[h!]
  \centering
  \includegraphics[width=0.7\linewidth]{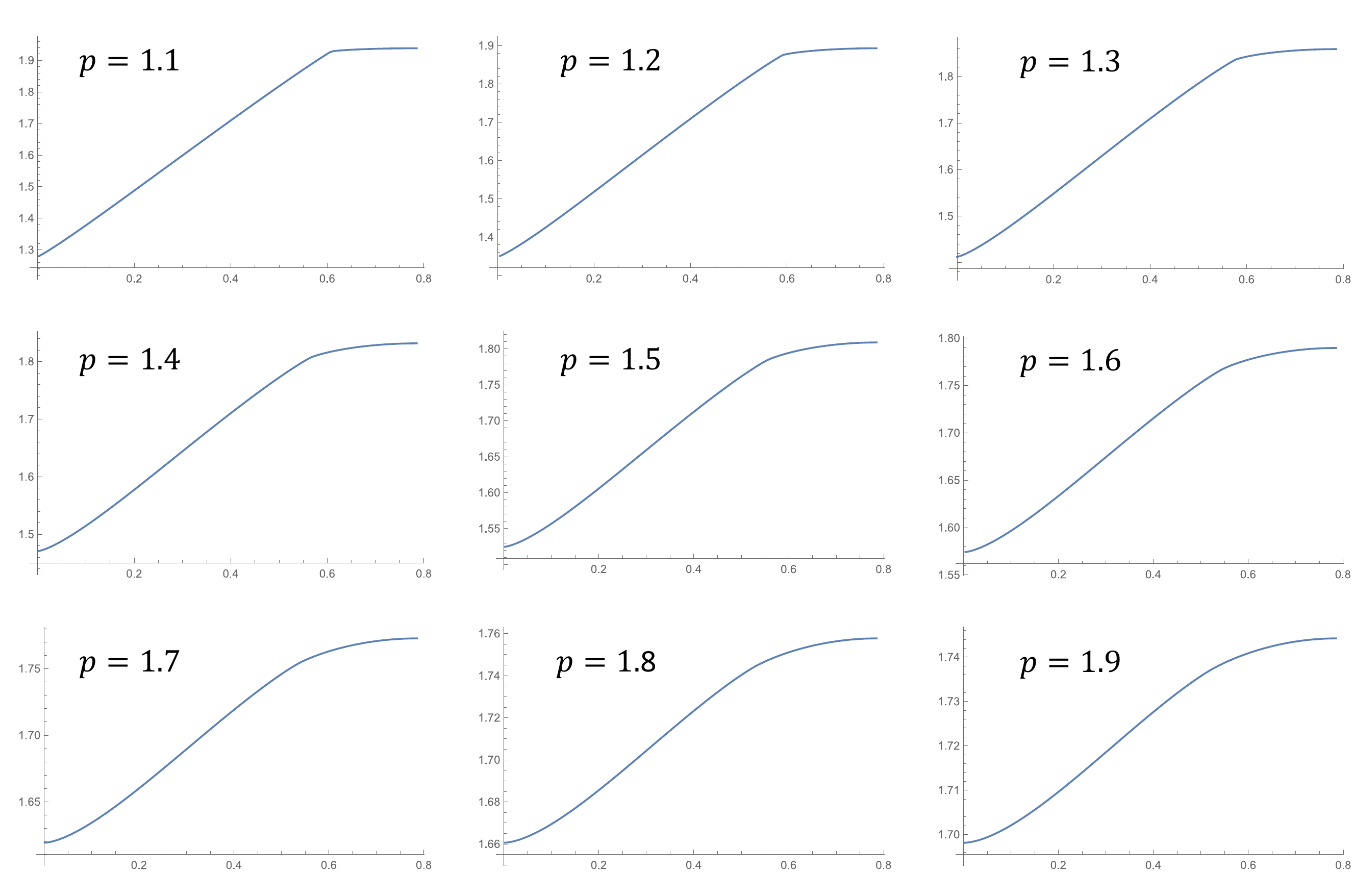}
\caption{
Figures depict increasing functions $\sigma_p(\theta)$, when $\theta \in [0,\pi/4]$, for various values of $p\in (1,2)$. 
}
\label{fig: SigmpaP2Dx6small}
\end{figure}

\begin{figure}[h!]
  \centering
  \includegraphics[width=0.7\linewidth]{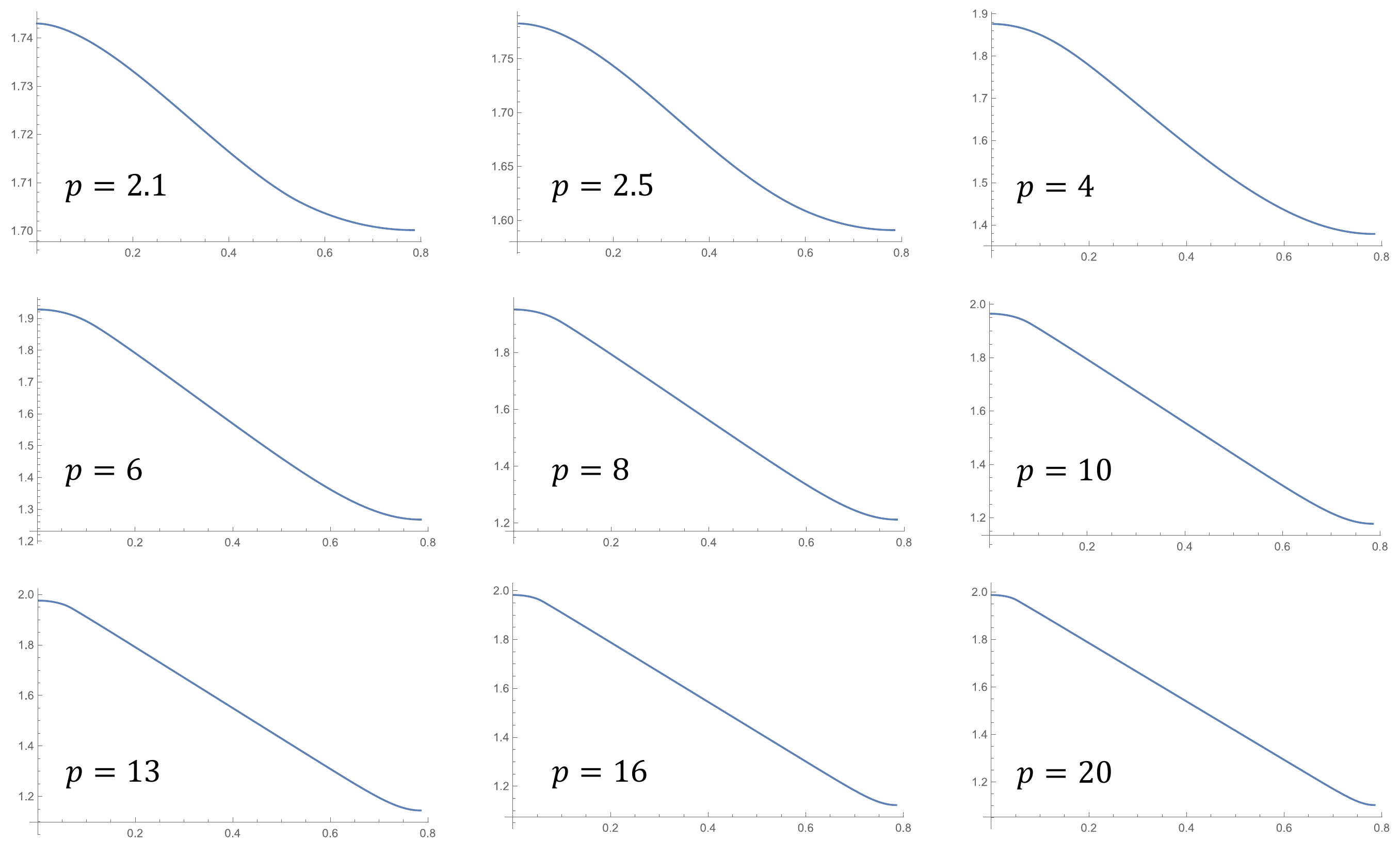}
\caption{
Figures depict decreasing functions $\sigma_p(\theta)$, when $\theta \in [0,\pi/4]$, for various values of $p\in (2,20]$. 
}
\label{fig: SigmpaP2Dx6large}
\end{figure}

We conclude the section by giving the technical details as to how computer-assisted numerical calculations can verify Lemma~\ref{lem: min arc is what I used restate} (and Lemma~\ref{lem: min arc is what I used restate x2}), and how the figures for $\sigma_p$ were produced. 
The reader may refer to Figure~\ref{fig: MinArcOrientation}. 
\begin{figure}[h!]
  \centering
  \includegraphics[width=0.7\linewidth]{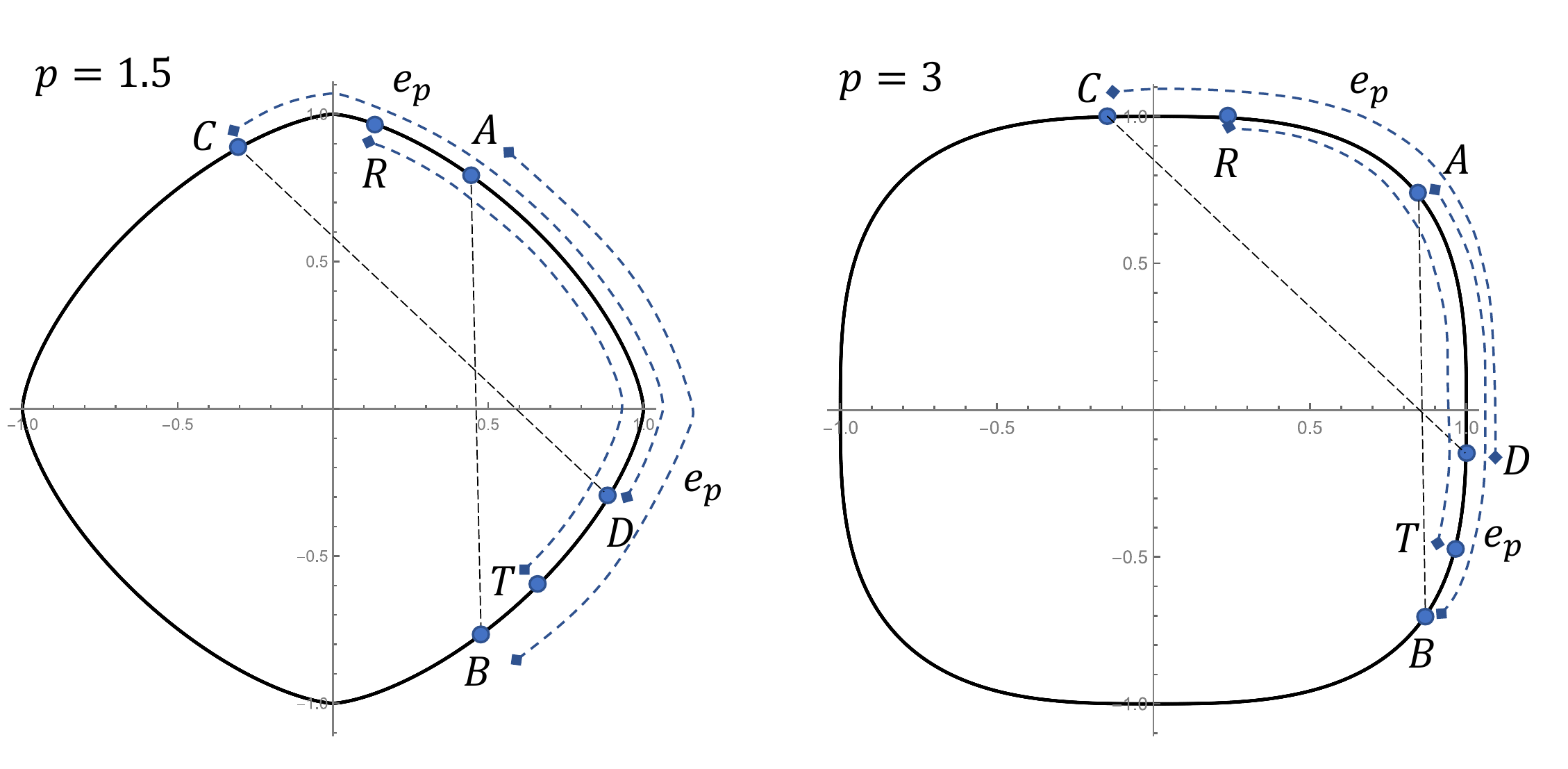}
\caption{
Two unit circles $\mathcal C_p$, for $p=1.5$ left, and $p=3$ right. For $p\in (1,2)$, arc $\arccc{BA}$ (of some fixed length) induces chord $AB$ of smallest length, and arc $\arccc{DC}$ (of the same arc length) induces chord $CD$ of largest length. 
For $p\in (2,\infty)$, arc $\arccc{BA}$ (of some fixed length) induces chord $AB$ of largest length, and arc $\arccc{DC}$ (of the same arc length) induces chord $CD$ of smallest length. 
}
\label{fig: MinArcOrientation}
\end{figure}
For each $p\geq 1$, we explain how we can calculate $\sigma_p(\theta)$ for $0\leq \theta \leq \pi/4$.

First we find points $A,B$ on $\mathcal C_p$ such that $\arccc{BA}$ has length $e_p$ and tangential angle $0$. For this, we set $A=(x,y)$ so that $B=(x,-y)$. Using the parametric equation $\bar r_p(t)=\left( (1-|t|^p)^{1/p},t \right)$, $t\in [-1,1]$, that describes $\mathcal C_p$ in the 4th and 1st quadrant, we find $w_a\geq0$ that is the solution to 
$$
\int_{-w}^w \norm{\bar{r_p}'(s)}_p \ddd s =e_p
$$
Therefore, $A=\left( (1-w_a^p)^{1/p},w_a \right)$.

Next we find first points $C,D$ on $\mathcal C_p$ such that $\arccc{DC}$ has length $e_p$ and tangential angle $\pi/4$. For this, we set $C=(x,y)$ so that $D=(y,x)$. We us parametric equation $r_p(t)=\left( t,(1-|t|^p)^{1/p} \right)$, $t\in [-1,1]$,that describes $\mathcal C_p$ in the 1st and 2nd quadrant. 
We observe that if $e_p > \pi_p/2$, then $C$ lies in the 2nd quadrant, and if $e_p<\pi_p/2$, then $C$ lies in the 1st quadrant. Therefore, we need to find $w_c$ satisfying 
\begin{align*}
2\int_{w}^0 \norm{r_p'(s)}_p \ddd s +\pi_p/2=e_p &\textrm{, if } e_p > \pi_p/2 \\
2\int_{w}^{2^{-1/p}} \norm{r_p'(s)}_p \ddd s =e_p &\textrm{, if } e_p < \pi_p/2 
\end{align*}
Then, we have that $C=\left( w_c,(1-|w_c|^p)^{1/p} \right)$, so that $C$ lies in the 2nd quadrant if $w_c<0$ and in the 1st quadrant if $w_c>0$.

Now we need to consider arbitrary point $R$ in the arc $\arccc{AC}$, and for each such point, find $T$ in the arc $\arccc{BD}$ satisfying $\mu_p(\arccc{TR})=e_p$. 
In particular, if $R=(x_R,y_R)$, then $w_c \leq x_r \leq (1-w_a^p)^{1/p}$, where $x_R = (1-w_a^p)^{1/p}$ gives tangential angle $0$ and $x_R=w_c$ gives tangential angle $\pi/4$. 
To conclude, for each $x_r \in [w_c , (1-w_a^p)^{1/p}]$ we find point $T$, and return chord length $\norm{RT}_p$. 
Figures~\ref{fig: SigmpaP3Dx2},\ref{fig: SigmpaP2Dx6small},\ref{fig: SigmpaP2Dx6large} depict $\norm{RT}_p$ (y-axis) as a function of $\left( 1-\frac{(x_r-w_c)}{(1-w_a^p)^{1/p}-w_c}\right)\frac{\pi}{4}$, and as $x_r$ ranges in $[w_c , (1-w_a^p)^{1/p}]$. 
This corresponds exactly to the plot of $\norm{RT}_p$, i.e. to $\sigma_p(\theta)$, as a function of the tangential angle $\theta$ of $\arccc{TR}$, only that the $x$-axis corresponding to the tangential angle is stretched according to our transformations. 
Overall the plots verify that $\sigma_p:[0,\pi/4]\mapsto \reals$ is increasing when $p\in [1,2)$ and decreasing when $p \in (2,\infty)$.

\section{Discussion}
We provided tight upper and lower bounds for the evacuation problem of two searchers in the wireless model from the unit circle in $\ell_p$ metric spaces, $p\geq 1$. This is just a starting point of revisiting well studied search and evacuation problems in general metric spaces that do not enjoy the symmetry of the Euclidean space. In light of the technicalities involved in the current manuscript, we anticipate that the pursuit of the aforementioned open problems will also give rise to new insights in convex geometry and computational geometry.

\bibliographystyle{plain}

\bibliography{ellp-refs}

\appendix

\section{Proof of Theorem~\ref{thm: explored optimal algo}}
\label{sec: proof of positive results}

Next we generalize the proof ideas of Lemma~\ref{lem: sym trajectories}. More specifically, we present a useful implication of Lemma~\ref{lem: sym trajectories} that will be invoked repeatedly in our analysis and that pertains to the distance of the two robots while they are searching for the exit.

\begin{lemma}
\label{equa: distance robots}
For $p \in (1,\infty)$, consider an execution of Algorithm Wireless-Search$_p$($\phi$), and let $(x_\tau,y_\tau)$ be the position of robot \#1 at time $\tau$. Then, we have
\begin{align*}
\delta_{p,0}(\tau) & = 2|y_\tau|, \\
\delta_{p,\pi/4}(\tau) & = 2^{1/p}|x_\tau-y_\tau|.
\end{align*}
Moreover, for $\phi \in \{0,\pi/4\}$, function $\delta_{p,\phi}(\tau)$ is strictly increasing when $\tau \in [0,\pi_p/2]$ and strictly decreasing when $\tau \in [\pi_p/2, \pi_p]$.
\end{lemma}

\begin{proof}
In the execution of Algorithm Wireless-Search$_p$($\phi$), suppose that robot \#1 follows trajectory $\rho_p(\phi+t)$, and let $(x_\tau,y_\tau)$ be its position after robots have searched the perimeter for time $\tau$. 
By Lemma~\ref{lem: sym trajectories}, the other robot is located at point $\rho_p(\phi-t)$. 

In particular, when $\phi=0$, the location of the two robots are $(x_\tau,y_\tau)$ and $(x_\tau,-y_\tau)$. As a result their $\ell_p$ distance is equal to 
$$
\delta_{p,0}(\tau) = \left( |x_\tau - x_\tau|^p + |y_\tau + y_\tau|^p\right)^{1/p}=
2|y_\tau|.
$$
When $\phi=\pi/4$, the location of the two robots are $(x_\tau,y_\tau)$ and $(y_\tau,x_\tau)$. As a result their $\ell_p$ distance is equal to 
$$
\delta_{p,\pi/4}(\tau) = \left( |x_\tau - y_\tau|^p + |y_\tau - x_\tau|^p\right)^{1/p}=
2^{1/p}|x_\tau - y_\tau|.
$$

Next we focus on the case that $\tau \in [0 ,\pi/2]$. The trajectory of robot \#1 can be alternatively described by parameterization~\eqref{equa: alternative param}. 
As a result, robots' distance can be described by some function $\bar{\delta}_{p,0}(s)$ on $s\in [-1,1]$. Moreover, for every $\tau \in [0 ,\pi/2]$ there exists unique $s=s(\tau)$ such that $(x_\tau,y_\tau)=r_p(s)$. Showing that $\bar{\delta}_{p,0}(s)$ is strictly increasing in $\tau \in [0 ,\pi/2]$, we can calculate (using the chain rule) 
$$
\frac{\partial}{\partial \tau} \bar{\delta}_{p,\phi}(s(\tau))
=
\bar{\delta}_{p,\phi}'(s(\tau)) \cdot s'(\tau).
$$
Clearly $s(\tau)$ is increasing in $\tau \in [0,\pi_p/2]$, and hence $s'(\tau)>0$. So the main claim of the lemma that robots' distances are strictly increasing in $\tau \in [0,\pi_p/2]$ follows by showing that 
$\frac{\partial}{\partial s} \bar{\delta}_{p,\phi}'(s) >0$ .

When $\phi=0$ robot \#1 moves along $r_p(s) = \left( -s, \left(1- |s|^p\right)^{1/p} \right)$, where $-1\leq s \leq 0$ (since $r_p(0)= \rho_p(\pi/2)$, a position that is reached after searching for time $\pi_p/2$). But then, 
$$
\bar{\delta}_{p,0}(s) = 2 \left(1- |s|^p\right)^{1/p} = 2 \left(1- (-s)^p\right)^{1/p}.
$$
Hence, 
$$
\bar{\delta}_{p,0}'(s) = 2 \left(1-(-s)^p\right)^{\frac{1}{p}-1} (-s)^{p-1} >0
$$
for all $s\in (-1,0)$ as wanted.

When $\phi=\pi/4$, robot \#1 moves along $r_p(s) = \left( -s, \left(1- |s|^p\right)^{1/p} \right)$, where $-2^{-1/p}\leq s \leq 2^{-1/p}$ (since $r_p(-2^{-1/p})= \rho_p(\pi/4)$ and $r_p(2^{-1/p})= \rho_p(3\pi/4)$ and the latter position is reached after searching for time $\pi_p/2$). 
Note also, that in this case, $\delta_{p,\pi/4}(\tau)= 2^{1/p}(y_\tau- x_\tau)$, and hence 
\begin{equation}
\label{equa: distance pi/4}
\bar{\delta}_{p,\pi/4}(s) = 2^{1/p}\left( 
 \left(1- |s|^p\right)^{1/p} +s
 \right).
\end{equation}
We distinguish two cases in order to compute $\bar{\delta}_{p,\phi}'(s)$. 
First, when $-2^{-1/p}\leq s \leq 0$, we have 
$$
\bar{\delta}_{p,\pi/4}'(s) = 
 2^{1/p}
 \left( 
\left(1- (-s)^p\right)^{1/p}+s
\right)'
=
 2^{1/p}
\left( 
(-s)^{p-1} \left(1-(-s)^p\right)^{\frac{1}{p}-1}+1
\right)>0.
$$
Second, when $0\leq s \leq 2^{-1/p}$, we have
$$
\bar{\delta}_{p,\pi/4}'(s) = 
 2^{1/p}
 \left( 
\left(1- (s)^p\right)^{1/p}+s
\right)'
=
 2^{1/p}
\left( 
1-s^{p-1} \left(1-s^p\right)^{\frac{1}{p}-1}
\right).
$$
Elementary algebraic calculations show that $1\geq s^{p-1} \left(1-s^p\right)^{\frac{1}{p}-1}$ exactly when $s \leq 2^{-1/p}$, and equality holds if $s = 2^{1/p}$.
We conclude that $\delta_{p,\phi}(\tau)$ is strictly increasing when $\tau \in [0,\pi_p/2]$ as promised. The fact that $\delta_{p,\phi}(\tau)$ is strictly decreasing when $\tau \in [0,\pi_p/2]$ is immediate from Lemma~\ref{lem: symmetries}.
\qed \end{proof}
It is interesting to note that $\delta_{p,\phi}(\tau)$ does not admit, in general, nice representations, and in fact calculating their values even for certain values of $\tau$ (and for arbitrary $p,\phi$) require numerical solutions of highly technical non linear equations. 
Next we provide worst case analysis of Algorithm~\ref{algo: wireless search phi} when $\phi \in \{ 0, \pi/4 \}$, that is we determine $E_{p,\phi} = \max_{\tau \in [0,\pi_p]} \mathcal E_{p,\phi} (\tau)$. For this, we take advantage of Lemma~\ref{equa: distance robots}, according to which $\mathcal E_{p,\phi} (\tau)$ is increasing when $\tau \in [0,\pi_p/2]$ for both $\phi=0,\pi/4$.\footnote{We believe that 
$\mathcal E_{p,\phi} (\tau)$
is increasing in $\tau \in [0,\pi_p/2]$ for all $\phi\in[0,\pi/4]$, even though that would be hard to prove. Nevertheless, such algorithms will not be optimal, and hence this property, even if true, is irrelevant to our analysis.} As a result, we will look for maximizers in $\tau \in [\pi_p/2,\pi_p]$.

\begin{lemma}
\label{lem: performance phi=0}
For $p\in (1,2]$, set $s_p = \left(\left(2^p-1\right)^{\frac{1}{p-1}}+1\right)^{-1/p}$. Then, we have 
$$
E_{p,0}
=
1+\pi_p/2+\int_{0}^{s_p} \left(z^{p^2-p} \left(1-z^p\right)^{1-p}+1\right)^{1/p} \ddd z
+ 
2 (1-s_p^p)^{1/p}. 
$$
\end{lemma}

\begin{proof}
By Lemma~\ref{equa: distance robots}, the evacuation time of Wireless-Search$_p$(0) is maximized when the exit is reported when robot \#1 is at location $r_p(s) = \left( -s, \left(1- s^p\right)^{1/p} \right),
$ for some $s\in [0,1]$, that is, after each robot has searched $\pi_p/2$ part of the unit circle. Clearly, robots have spent time $\pi_p/2$ searching till they reach $r_p(0)$. They also need additional time $\int_0^s \norm{r_p'(z)}_p \ddd z$ till the exit is reported, at which time their distance, as per Lemma~\ref{equa: distance robots}, equals $2\left(1- s^p\right)^{1/p}$. Overall, the evacuation cost in this case is described by the following function 
$$
f(s) = 1+\pi_p/2 + \int_0^s \norm{r_p'(z)}_p \ddd z +2\left(1- s^p\right)^{1/p},
$$ 
where $s\in [0,1]$. 
The proof of our main claim follows by technical Lemma~\ref{lem: maximizer p12} that shows that $f(s)$ is indeed maximized at $s=s_p$. 
\qed \end{proof}

\begin{lemma}
\label{lem: maximizer p12}
Function 
$
f(s) = 1+\pi_p/2 + \int_0^s \norm{r_p'(z)}_p \ddd z +2\left(1- s^p\right)^{1/p},
$
over $s\in [0,1]$ is maximized at $s_p=\left(\left(2^p-1\right)^{\frac{1}{p-1}}+1\right)^{-1/p}$.
\end{lemma}

\begin{proof}
Note that $f(0) = 1+\pi_p/2+2$, and that $f(1)=1+\pi_p$ (where $\pi_p\leq 4$). Hence, the maximum is either $3+\pi_p$, or it is attained at some critical point of $f(s)$. Indeed, we verify next that $s=s_p$ is the only value in $[0,1]$ which is a root to $f'(s)=0$. For this recall that since $s\in [0,1]$ we have $r_p(s) = \left( -s, \left(1- s^p\right)^{1/p} \right)$. 
Indeed, by the Fundamental Theorem of Calculus, we have 
\begin{align*}
f'(s) &=
\norm{r_p'(s)}_p + \frac{\partial}{\partial s} 2\left(1- s^p\right)^{1/p} \\
& =
\left( 1 + \frac{s^{p(p-1)}}{\left(1-s^p\right)^{p-1}} \right)^{1/p}
-
2 \frac{s^{p-1}}{\left(1-s^p\right)^{1-1/p}}.
\end{align*}
Since $s\in [0,1]$ is follows that $f'(s)=0$ exactly when 
\begin{align*}
& 1 + \frac{s^{p(p-1)}}{\left(1-s^p\right)^{p-1}} 
= 2^p \frac{s^{p(p-1)}}{\left(1-s^p\right)^{p-1}}. \\
\Leftrightarrow  &
\frac{s^{p(p-1)}}{\left(1-s^p\right)^{p-1}}
= (2^p-1)^{-1} \\
\Leftrightarrow  &
\frac{s^{p}}{1-s^p}
= (2^p-1)^{-p+1} \\
\Leftrightarrow  &
s^p
= \left( 
\left(2^p-1\right)^{\frac{1}{p-1}}+1 \right)^{-1}\\
\Leftrightarrow  &
s
= \left(\left(2^p-1\right)^{\frac{1}{p-1}}+1\right)^{-1/p}.
\end{align*}
In other words, $s_p$ is the unique critical point to $f(s)$. 
Finally, to see that $s_p$ is indeed a maximizer, note that $\lim_{s\rightarrow 0^+} f'(s)=1>0$ and that $\lim_{s\rightarrow 1^-} f'(s) = - \infty$. We conclude that $f(s)$ is strictly increasing at $s\rightarrow 0^+$ and strictly decreasing at $s\rightarrow 1^-$. Since moreover $f'(s)$ has a unique root in $[0,1]$, it follows that $f'(s)$ is strictly concave in $[0,1]$ , and hence any root of $f'(s)$ in the same interval is a maximizer of $f(s)$. 
\qed \end{proof}

\ignore{
\begin{corollary}
For all $p \in (1,2]$, the critical chord distance that induces maximum evacuation time is equal to 
$$
2 (1-s_p^p)^{1/p} =
\frac{2}{\left( \left(2^p-1\right)^{\frac{1}{1-p}}+1\right)^{1/p}} 
$$
We need to show that among all chords with length the value above, the one that maximizes the corresponding arc is when parallel to x=0 or y=0. 

Similarly, the unexplored arc has length 
$$
2\int_{s_p}^1 \norm{r_p'(z)}_p \ddd z
=
2\int_{s_p}^1 \left(z^{p^2-p} \left(1-z^p\right)^{1-p}+1\right)^{1/p} \ddd z.
$$
We need to show that among all arcs (less than $\pi_p$) with length as above, the one that minimizes the corresponding chord is when x=0 or y=0 are bisectors of the arc (or when chord is parallel to x=0 or y=0). 
\end{corollary}
}

\begin{lemma}
\label{lem: performance phi=pi_p/4}
For $p\in [2,\infty)$, let $w_p$ be the unique\footnote{See Lemma~\ref{lem: critical p2infty} and its proof.} root
 to equation $w^p+1=2(1-w)^p$.
Let also $s_p = \left(w_p^{p/(p-1)} +1\right)^{-1/p}$. 
Then, we have that
$$
E_{p,\pi/4}=
1+\pi_p/2+\int_{2^{-1/p}}^{s_p} \left(z^{p^2-p} \left(1-z^p\right)^{1-p}+1\right)^{1/p} \ddd z
+ 
2^{1/p} \left( 
\left(1-s_p^p\right)^{1/p}
+s_p
\right),
$$
or $E_{p,\pi/4}=1+\pi_p$. 
\end{lemma}

\begin{proof}

By Lemma~\ref{equa: distance robots}, the evacuation time of Wireless-Search$_p$($\pi/4$) is maximized when the exit is reported when robot \#1 is at location $\rho_p(t)$ for some $t \in [3\pi/4, 5\pi/4]$. We examine separately the cases $t \in [3\pi/4, \pi]$ and $t \in [\pi, 5\pi/4]$.

First, we restrict the analysis to $t \in [3\pi/4, \pi]$.
The location of robot \#1 is given by $r_p(s) = \left( -s, \left(1- s^p\right)^{1/p} \right)$ for some $s\in [2^{-1/p},1]$, in which interval the exit is reported. 
Clearly, robots have spent time $\pi_p/2$ searching till they reach $r_p(2^{-1/p})$.
They also need additional time $\int_{2^{-1/p}}^s \norm{r_p'(z)}_p \ddd z$ till the exit is reported, at which time their distance, as per Lemma~\ref{equa: distance robots}, equals 
$2^{1/p}\left( 
 \left(1- s^p\right)^{1/p} +s
 \right)$
 (see~\eqref{equa: distance pi/4}, and recall that $s\geq 0$). 
 Overall, the evacuation cost in this case is described by the function 
$$
f_1(s) = 1+\pi_p/2 + \int_{2^{-1/p}}^s \norm{r_p'(z)}_p \ddd z +2^{1/p}\left( 
 \left(1- s^p\right)^{1/p} +s
 \right),
$$ 
where $s\in [2^{-1/p},1]$. 
In technical Lemma~\ref{lem: critical p2infty}, we show that $s=s_p$ is the unique critical point of $f(s)$ in $[2^{-1/p},1]$, and hence the unique candidate maximizer in the same interval.

Second, we restrict the analysis to $t \in [\pi, 5\pi/4]$.
Our main claim in this case is that the corresponding evacuation cost function has no critical point, and the lemma will follow. 
In order to calculate the evacuation cost function, we still use parameterization~\ref{equa: alternative param}, which however cannot describe the location of robot \#1 (which is now moving in the 3rd quadrant). For this we will rely on the lemmata we already introduced pertaining to the symmetries of $\mathcal C_p$. 

Clearly, robots have spent time $3\pi_p/4$ searching till they reach point 
$\rho_p(\pi)=r_p(1)=(-1,0)$.
Utilizing Lemma~\ref{lem: symmetric measures} (and the point of symmetry, as per Lemma~\ref{lem: symmetries}), robots also need additional time $\int_{-1}^s \norm{r_p'(z)}_p \ddd z$ till the exit is reported, for some $s\in [-1,-2^{-1/p}]$. At this time, their distance, as per Lemma~\ref{equa: distance robots} (note that in this case $x_\tau \geq y_\tau$), equals 
$2^{1/p}\left( 
 - \left(1- (-s)^p\right)^{1/p} -s
 \right)$. 
 Overall, the evacuation cost in this case is described by function 
$$
f_2(s) = 1+3\pi_p/4 + \int_{-1}^s \norm{r_p'(z)}_p \ddd z 
-
2^{1/p}\left( 
 \left(1- (-s)^p\right)^{1/p} +s
 \right).
$$ 
where $s\in [-1,2^{-1/p}]$. 
In technical Lemma~\ref{lem: no critical p2infty}, we show that $f_2(s)$ has no critical points when $s\leq 0$, and hence no critical points when $s\in [-1,2^{-1/p}]$.

Overall, we showed that after time $\pi_p/2$ of searching, the evacuation cost function has a unique critical point with respect to time. Since the evacuation cost was increasing in the first $\pi_p/2$ time of searching, it follows that the critical point is a maximizer, unless it is a saddle point, in which case the worst case cost is attained at the end of the search, that is in case the cost is $1+\pi_p$. 
\qed \end{proof}

\begin{lemma}
\label{lem: critical p2infty}
For $p\in [2,\infty)$, let $w_p$ be the unique root
 to equation $w^p+1=2(1-w)^p$. 
Then, $s_p = \left(w_p^{p/(p-1)} +1\right)^{-1/p}$ is the unique critical point of function 
$$
f_1(s) = 1+\pi_p/2 + \int_{2^{-1/p}}^s \norm{r_p'(z)}_p \ddd z 
+
2^{1/p}\left( 
 \left(1- s^p\right)^{1/p} +s
 \right).
$$ 
when $s\in [2^{-1/p},1]$. 
\end{lemma}

\begin{proof}
Using the Fundamental Theorem of Calculus, we have 
\begin{align*}
f_1'(s) &=
\norm{r_p'(s)}_p + 
2^{1/p}
\frac{\partial}{\partial s} \left( 
 \left(1- s^p\right)^{1/p} +s
 \right) \\
& =
\left( 1 + \frac{s^{p(p-1)}}{\left(1-s^p\right)^{p-1}} \right)^{1/p}
+
2^{1/p}\left( 
- \frac{s^{p-1}}{\left(1-s^p\right)^{1-1/p}} +1
 \right).
\end{align*}
Set $w:=\frac{s^{p-1}}{\left(1-s^p\right)^{1-1/p}}$, and note that a critical point $s=s(w)$ must satisfy
$$
w^p+1=2(1-w)^p.
$$
The latter equation has a unique solution $w_p\in (0,1)$. 
To see why, define $a(w):=w^p+1-2(1-w)^p$, and note that $a(0)=-1$ and $a(1)=2$. Moreover $a(w)$ is clearly strictly increasing in $a\in (0,1)$, so indeed $a(w)$ has a unique root in $(0,1)$.

Now solving expression $w=\frac{s^{p-1}}{\left(1-s^p\right)^{1-1/p}}$ for $s$ gives the unique solution $s(w)=\left(w_p^{p/(p-1)} +1\right)^{-1/p}$. Some straightforward calculations then show $w\geq 0$ implies that $s(w)\leq 1$ and that $w\leq 1$ implies that $s(w) \geq 2^{-1/p}$, as wanted. 
\qed \end{proof}

\begin{lemma}
\label{lem: no critical p2infty}
Function 
$$
f_2(s) = 1+3\pi_p/4 + \int_{-1}^s \norm{r_p'(z)}_p \ddd z 
-
2^{1/p}\left( 
 \left(1- (-s)^p\right)^{1/p} +s
 \right).
$$ 
has no critical points when $s\leq 0$. 
\end{lemma}

\begin{proof}
Recall that $s\leq 0$ and that $r_p(s)=\left( -s, \left(1- |s|^p\right)^{1/p} \right)
=\left( -s, \left(1- (-s)^p\right)^{1/p} \right)$. 
Using the Fundamental Theorem of Calculus, we have 
\begin{align*}
f_2'(s) &=
\norm{r_p'(s)}_p 
-
2^{1/p}
\frac{\partial}{\partial s} 
\left( 
 \left(1- (-s)^p\right)^{1/p} +s
 \right)
  \\
& =
\left( 1 + \frac{(-s)^{p(p-1)}}{\left(1-(-s)^p\right)^{p-1}} \right)^{1/p}
-
2^{1/p}\left( 
\frac{(-s)^{p-1}}{\left(1-(-s)^p\right)^{1-1/p}} +1
 \right).
\end{align*}
Set $w:=\frac{(-s)^{p-1}}{\left(1-(-s)^p\right)^{1-1/p}}$, and note for $s\leq 0$ we have $w\geq 0$. At the same time a critical point $s=s(w)$ must satisfy
$$
w^p+1=2(1+w)^p.
$$
However, the latter equation has no non-negative root. 
To see why, define $b(w):=2(1+w)^p - w^p-1$, and note that $b(0) =1>0$. 
Moreover, when $w>0$ we have
$$
b'(w) 
= 2p (1+w)^{p-1}-p w^{p-1} 
= p(2(1+w)^{p-1}-w^{p-1}) 
\geq p (2w^{p-1}-w^{p-1})  >0,
$$
hence $b(w)$ is strictly increasing. As a result, it cannot have a root in $(0,1)$. 
\qed \end{proof}

\ignore{
\begin{corollary}
For all $p \in [2,\infty]$, the critical chord distance that induces maximum evacuation time is equal to 
$$
2^{1/p} \left( 
\left(1-s_p^p\right)^{1/p}
+s_p
\right),
 $$
where $s_p = \left(w_p^{p/(p-1)} +1\right)^{-1/p}$, 
and $w_p$ is the unique root
 to equation $w^p+1=2(1-w)^p$.

We need to show that among all chords with length the value above, the one that maximizes the corresponding arc is when parallel to y=x or y=-x. 

Similarly, the unexplored arc has length 
$$
\pi_p/2+
2\int_{s_p}^1 \norm{r_p'(z)}_p \ddd z
=
\pi_p/2+
2 \int_{s_p}^1 \left(z^{p^2-p} \left(1-z^p\right)^{1-p}+1\right)^{1/p} \ddd z
$$
We need to show that among all arcs (less than $\pi_p$) with length as above, the one that minimizes the corresponding chord is when y=x or y=-x are bisectors of the arc (or when chord is parallel to y=x or y=-x). 
\end{corollary}
}

\end{document}